\documentclass[letterpaper,journal,twoside]{IEEEtran}
\usepackage{amsmath,amsfonts,amsthm,amssymb}
\usepackage{algorithmic}
\usepackage{algorithm}
\usepackage{array}
\usepackage[caption=false,font=footnotesize,labelfont=rm,textfont=rm]{subfig}
\usepackage{textcomp}
\usepackage{stfloats}
\usepackage{url}
\usepackage{xurl}
\usepackage{verbatim}
\usepackage{graphicx}
\usepackage{microtype}
\usepackage[nobreak]{cite}
\usepackage{bm}
\usepackage{color}
\usepackage{booktabs}
\usepackage{multirow}
\usepackage{tikz}
\usepackage{pgfplots}
\usetikzlibrary{positioning,arrows.meta}
\pgfplotsset{compat=1.18}
\newtheoremstyle{compact}{3pt}{3pt}{\itshape}{}{\bfseries}{.}{0.5em}{}
\theoremstyle{compact}

\newtheorem{proposition}{Proposition}[section]

\renewcommand{\thetheorem}{\arabic{section}.\arabic{theorem}}
\renewcommand{\theremark}{\arabic{section}.\arabic{remark}}
\renewcommand{\theproposition}{\arabic{section}.\arabic{proposition}}
\renewcommand{\theassumption}{\arabic{section}.\arabic{assumption}}
\renewcommand{\thelemma}{\arabic{section}.\arabic{lemma}}
\renewcommand{\thecorollary}{\arabic{section}.\arabic{corollary}}

\usepackage[hidelinks]{hyperref}
\AtBeginDocument{\binoppenalty=10000\relpenalty=10000}

\begin{document}

\title{A Framework for Learning and Reliability Certification in Neural Polar Decoding}

\author{Hongzhi Zhu,~ 
        % Wei Xu\IEEEauthorrefmark{1},~\IEEEmembership{Fellow,~IEEE,}~
        % and Xiaohu You\IEEEauthorrefmark{1},~\IEEEmembership{Fellow,~IEEE,}
        Wei Xu,~\IEEEmembership{Fellow,~IEEE,}~
        and Xiaohu You,~\IEEEmembership{Fellow,~IEEE}
        % <-this % stops a space
%
\thanks{H. Zhu is with the National Mobile Communications Research Laboratory~\mbox{(NCRL)}, Southeast University, Nanjing, 210096, and the Purple Mountain Laboratories, Nanjing, 210096, China (email: hz\_zhu@seu.edu.cn).}%
\thanks{W. Xu and X. You are with the National Mobile Communications Research Laboratory~\mbox{(NCRL)}, Southeast University, Nanjing, 210096, China, and the Purple Mountain Laboratories, Nanjing, 210096, China (email: \{wxu, xhyu\}@seu.edu.cn). X. You is also the corresponding author of this paper.}
}

% The paper headers
\markboth{}%
{H. Zhu \MakeLowercase{\textit{et al.}}}

\IEEEaftertitletext{\vspace{-36pt}}
\maketitle
\setlength{\parskip}{0pt}
\setlength{\abovedisplayskip}{0pt plus 1pt minus 0pt}
\setlength{\belowdisplayskip}{0pt plus 1pt minus 0pt}
\setlength{\belowdisplayshortskip}{0pt plus 1pt minus 0pt}

\begin{abstract}
This paper investigates a framework that connects network training, independently certified population risk and sequential reliability in neural polar decoding. Network outputs determine information-bit decisions by their signs. Posterior mean-squared error~\mbox{(MSE)} controls bit-error rate~\mbox{(BER)}, and a first-error identity connects errors evaluated with true preceding bits to sequential block-error rate~\mbox{(BLER)}. Three methods exploit sequential decoding structure to certify changes to trained receivers. First, {distances from trained outputs to the zero decision threshold}, termed learned margins, yield reliability bounds for over-parameterized neural network~\mbox{(ONN)} updates subject to separate mean-squared output-change constraints at information positions. Second, a shared loss takes the largest {block-error indicator difference} between candidate decoders and successive cancellation~\mbox{(SC)} on each frame. Its expectation bounds the corresponding BLER differences, termed excess BLER, for all fixed combinations of information-position implementations. An exact expression quantifies the gap between this shared bound and the largest BLER difference among fixed decoders. Third, verified linear programming~\mbox{(LP)} and candidate-set subdivision tighten these guarantees with controllable computation at fixed confidence. ONN experiments assess conditional fitting and weight updates. Polar factor-graph experiments demonstrate reusable combination guarantees and storage reduction under prescribed excess-BLER allowances. Component timing identifies the dominant graph-propagation cost.
\end{abstract}


\vspace{-6pt}\section{Introduction}\vspace{-3pt}
Polar codes asymptotically achieve the capacity of symmetric binary-input memoryless channels~\mbox{\cite{ref5}}. Successive cancellation~\mbox{(SC)} decodes their synthesized bit-channels in sequence. Bitwise neural network~\mbox{(NN)} decoders use one predictor per information position, collectively termed a predictor bank. Reliability guarantees for these decoders require a unified theoretical framework that connects network training, population prediction error, and sequential decoding performance.

{Each predictor maps channel observations and preceding source bits to a signed decision score. Posterior mean-squared error~\mbox{(MSE)} controls bit-error rate~\mbox{(BER)}, and independent samples bound population risk. Transmitted-prefix evaluation uses true preceding bits. Before the first sequential error, preceding decoder decisions equal these bits. This identity connects local errors to block-error rate~\mbox{(BLER)} across architectures.}

\vspace{-9pt}\subsection{Conventional Polar Decoding}\vspace{-3pt}
SC follows a single decision path. Successive cancellation list~\mbox{(SCL)} decoding~\mbox{\cite{ref12.1,ref12.2}} retains multiple paths, stack decoding~\mbox{\cite{ref12.3}} explores them by priority, and cyclic redundancy check~\mbox{(CRC)}-aided SCL~\mbox{\cite{ref7}} adds an outer consistency test. These methods provide references for reliability and implementation cost.

\vspace{-9pt}\subsection{Deep Learning for Polar Decoding}\vspace{-3pt}
Learned decoders approximate decoding mappings or incorporate code-structured updates~\mbox{\cite{ref9.1}}: multilayer perceptrons for short codes~\mbox{\cite{ref12.5}}, neural SC and belief-propagation-flip decoding~\mbox{\cite{Ha2019deeppolar,ref15}}, or curricula that exploit nested polar subcodes~\mbox{\cite{hebbar2023crisp}}. {Reliability analysis must quantify generalization from finite training data and error propagation through preceding decisions. Population-risk bounds and a bitwise-to-sequential error relation address these effects. Quantization or predictor replacement further requires guarantees for the modified decoder.}

\vspace{-9pt}\subsection{Optimization and Generalization Theory}\vspace{-3pt}
{Optimization guarantees for over-parameterized neural networks~\mbox{(ONNs)} hold under specified conditions~\mbox{\cite{So2019theory,zhang2021ONN}}. Wide rectified linear unit~\mbox{(ReLU)} networks admit gradient descent~\mbox{(GD)} convergence results~\mbox{\cite{du2018gradient}} and separate generalization analyses~\mbox{\cite{cao2019generalization}}. Surrogate-risk and posterior-margin bounds~\mbox{\cite{bartlett2006risk,audibert2007}} relate regression error to error probability above \textit{maximum a posteriori}~\mbox{(MAP)} decoding. The posterior margin is the absolute conditional probability difference between bit values. Training-sample classification margins and network norms also yield generalization bounds~\mbox{\cite{bartlett2017margin}}. {Independent samples certify population risk. Distances between trained-network outputs and the zero decision threshold provide the margins for subsequent updates.}}

Quantization verification~\mbox{\cite{zhang2023qebverif}} bounds network-output differences over specified input regions. The \emph{Learn then Test} framework~\mbox{\cite{angelopoulos2025ltt}} uses hypothesis testing to select predictive procedures with statistical risk guarantees. {Sequential structure supports three certification methods. Mutually exclusive first-error events yield update bounds with separate constraints at each position. Local candidate errors yield a common bound for decoder combinations. Fixed implementation choices across frames permit computational refinement. These methods use established ONN optimization, surrogate-risk inequalities, convex duality and concentration to certify changes to trained receivers.}

\vspace{-9pt}\subsection{Contributions}\vspace{-3pt}
The main contributions are:
\begin{itemize}
\item {A first-error analysis bounds excess BLER, {the difference obtained by subtracting the SC BLER from the neural decoder's BLER under the same operating law}. Learned margins yield a uniform bound for updates with separate mean-squared output-change constraints at each information position. An exact moment-constrained optimization characterizes this bound and can improve upon a single total-budget constraint. Input second moments translate ONN weight changes into these budgets.}
\item A shared frame loss takes the largest neural-minus-SC block-error indicator difference on each frame. It certifies all decoders formed by selecting one implementation per information position, without enumerating them. Its exact population gap from the worst fixed decoder quantifies the conservatism of allowing different decoders to contribute errors on different frames. {The common guarantee supports subsequent precision and readout choices under an excess-BLER allowance, with one assessment reusable across feasible configurations.}
\item Verified linear programming~\mbox{(LP)} bounds and adaptive partitioning tighten the common guarantee at the prescribed confidence level. Partitioning subdivides the candidate set to control the computation--tightness tradeoff. Exact empirical BLER maximization over fixed combinations is nondeterministic polynomial-time~\mbox{(NP)}-hard for arbitrary predictor-error tables, which motivates these refinements.
\end{itemize}

\vspace{-9pt}\subsection{Paper Outline}\vspace{-3pt}
\mbox{Sections~\ref{sec_preliminaries}--\ref{sec_onn_instance}} present the models, reliability certificates and ONN optimization. \mbox{Sections~\ref{sec_experiments} and~\ref{sec_complexity}} report numerical results and implementation costs. \mbox{Section~\ref{sec_conclusion}} concludes the paper.

\vspace{-9pt}\subsection{Notation}\vspace{-3pt}
Let $[n]=\{1,\ldots,n\}$. Let $\operatorname{Unif}(\mathcal A)$ denote the uniform law on $\mathcal A$ and $\mathcal N(\bm\mu,\bm{\mathrm{\Sigma}})$ a Gaussian law. The subvector $\bm x[i:j]$ selects entries $i$ through $j$, and $\bm u^{i-1}=\bm u[1:i-1]$ denotes the source prefix. Also, $[\bm u,\bm v]$ denotes concatenation. For $\mathbf A=[\bm a_1,\ldots,\bm a_B]$, define $\|\mathbf A\|_{2,\infty}=\max_j\|\bm a_j\|_2$ and $\|\mathbf A\|_{\rm F}^2=\sum_j\|\bm a_j\|_2^2$, so $\|\mathbf A\|_{\rm F}\le\sqrt B\|\mathbf A\|_{2,\infty}$. The symbols $\mathbb I(\cdot)$, $\mathrm Q(\cdot)$, and $\mathcal O(\cdot)$ denote an indicator, the standard Gaussian upper-tail function, and asymptotic order, respectively.

\vspace{-6pt}\section{Polar Decoding and Generic Bitwise Prediction Model}\vspace{-3pt}
\label{sec_preliminaries}
\subsection{System Model and Polar Construction}\vspace{-3pt}
\label{polar_code}
For a polar code of length $N=2^n$ and rate $R=K/N$, let $\mathcal A_K\subset[N]$ contain the $K$ information positions and $\mathcal A_K^{\mathrm c}$ the frozen positions. A message $\bm m\in\{0,1\}^K$ is placed in $\bm u_{\mathcal A_K}=\bm m$, with $\bm u_{\mathcal A_K^{\mathrm c}}=\bm0$. Encoding is over the binary field~\mbox{\cite{ref5}}:
\begin{equation}
\bm x=\bm u\mathbf G_N,\qquad
\mathbf G_N=\mathbf B_N\mathbf F^{\otimes n},\qquad
\mathbf F=\begin{bmatrix}1&0\\1&1\end{bmatrix}.
\label{polar_encode}
\end{equation}
Here $\mathbf B_N$ is the bit-reversal permutation matrix and $\otimes$ denotes the Kronecker product. With bit reversal absorbed into consistent source and information-set indexing, the equivalent representation is $\bm x=\bm u\mathbf F^{\otimes n}$. We consider binary phase-shift keying~\mbox{(BPSK)}, which maps codeword bit $x_j$ to $s_j=1-2x_j$. Under additive white Gaussian noise~\mbox{(AWGN)},
\begin{equation}
\bm y=\bm s+\bm n,\quad
\bm n\sim\mathcal N(\bm0,\sigma_{\mathrm n}^2\mathbf I_N),\quad
\bm\ell=2\bm y/\sigma_{\mathrm n}^2.
\label{eq_awgn_channel_v2}
\end{equation}
For block-Rayleigh fading, one nonnegative amplitude $h$ with $\mathbb E[h^2]=1$ is drawn per codeword. With perfect channel state information~\mbox{(CSI)}, the observation and channel log-likelihood ratio~\mbox{(LLR)} are
\begin{equation}
y_j=hs_j+n_j,\qquad
\ell_j=2hy_j/\sigma_{\mathrm n}^2,\quad j\in[N].
\label{eq_rayleigh_channel_v2}
\end{equation}
All decoders receive the same channel sufficient statistic $\bm\ell$. Frozen bits are fixed, and information-bit predictions produce $\hat{\bm u}$ and the recovered message $\hat{\bm m}=\hat{\bm u}_{\mathcal A_K}$. \mbox{Fig.~\ref{fig_system}} shows this successive receiver. With unit-energy BPSK symbols, $\mathrm E_{\mathrm b}=1/R$ and $\mathrm N_0=2\sigma_{\mathrm n}^2$. Thus $\sigma_{\mathrm n}^2=[2R\,10^{\xi/10}]^{-1}$ at $\mathrm E_{\mathrm b}/\mathrm N_0=\xi$~decibels~\mbox{(dB)}.
\begin{figure}[!b]
\centering
\vspace*{8pt}
\resizebox{\columnwidth}{!}{\begin{tikzpicture}[font=\footnotesize,>=latex,line width=0.55pt]
\tikzset{systembox/.style={draw,rounded corners=1.5pt,text width=1.55cm,minimum height=0.85cm,align=center,inner sep=2pt},signal/.style={font=\footnotesize}}
\node[systembox] (source) {Message\\placement};
\node[systembox,right=0.40cm of source] (encoder) {Polar\\coding};
\node[systembox,right=0.40cm of encoder] (modulator) {Modulator};
\node[systembox,right=0.40cm of modulator] (channel) {Channel};
\node[systembox,below=0.65cm of channel] (llr) {LLR\\computation};
\node[systembox,below=0.65cm of modulator] (decoder) {Bitwise\\decoder};
\node[systembox,below=0.65cm of encoder] (extractor) {Message\\extraction};
\draw[->] ([xshift=-0.42cm]source.west) -- node[signal,above] {$\bm m$} (source.west);
\draw[->] (source.east) -- node[signal,above] {$\bm u$} (encoder.west);
\draw[->] (encoder.east) -- node[signal,above] {$\bm x$} (modulator.west);
\draw[->] (modulator.east) -- node[signal,above] {$\bm s$} (channel.west);
\draw[->] (channel.south) -- node[signal,left] {$\bm y$, CSI} (llr.north);
\draw[->] (llr.west) -- node[signal,above] {$\bm\ell$} (decoder.east);
\draw[->] (decoder.west) -- node[signal,above] {$\hat{\bm u}$} (extractor.east);
\draw[->] (extractor.west) -- node[signal,above] {$\hat{\bm m}$} ++(-0.42cm,0);
\tikzset{feedback/.style={draw=black,text=black}}
\draw[->,feedback] ([xshift=-0.30cm]decoder.south) -- ++(0,-0.43cm)
 -- node[signal,below] {Previous decisions} ([xshift=0.30cm,yshift=-0.43cm]decoder.south)
 -- ([xshift=0.30cm]decoder.south);
\draw[->,dashed] (source.south) -- ++(0,-0.30cm)
 -| node[signal,pos=0.65,above] {Frozen bits} (decoder.north);
\end{tikzpicture}
}
\caption{Polar-coded transmission and successive decoding using channel LLRs and preceding decisions.}
\label{fig_system}
\end{figure}

\vspace{-9pt}\subsection{General Bitwise Neural Predictor}\vspace{-3pt}
\label{bitwise_nn_model}
For $i\in\mathcal A_K$, let $F_i(\bm z_i;\bm\theta_i)\in\mathbb R$ be a trained predictor with parameter vector $\bm\theta_i$. Reference inputs contain transmitted preceding bits, whereas sequential inputs contain preceding decoder decisions:
\begin{equation}
\bm z_i^{\rm ref}=[\bm\ell,\bm u[1:i-1]],\qquad
\bm z_i^{\rm seq}=[\bm\ell,\hat{\bm u}[1:i-1]].
\label{eq_bitwise_input}
\end{equation}
Let $U_i\in\{0,1\}$ be the random source bit with realization $u_i$. Let $\mathcal D_i$ denote the joint operating law of input $\bm Z_i$ and signed information bit $V_i=2U_i-1$, including the message, channel, and prefix-generation rules. All risks in a comparison refer to this same law. Define the posterior target and its LLR by
\begin{gather}
r_i^\star(\bm z_i)=\mathbb E[V_i\mid\bm Z_i=\bm z_i],\label{eq_soft_target}\\
L_i(\bm z_i)=\log\frac{\mathbb P(U_i=0\mid\bm z_i)}{\mathbb P(U_i=1\mid\bm z_i)}.\label{eq_bitwise_llr}
\end{gather}
The conditional law includes the code\textquotesingle s frozen constraints. A standard SC recursion averages over future source bits, so its soft score is treated as a computable teacher unless these conditional laws coincide. Thus $r_i^\star=-\tanh(L_i/2)$, and the posterior margin is $\gamma_i(\bm z_i)=|r_i^\star(\bm z_i)|$. The MAP rule and the neural rule use the same zero-threshold convention:
\begin{equation}
g_i^\star=\mathbb I(r_i^\star\ge0),\qquad
\hat u_i=\mathbb I(F_i\ge0).
\label{eq_hard_decision_i}
\end{equation}
The population posterior-regression error is
\begin{equation}
E_{\mathcal D_i}(F_i)
=\mathbb E\!\left[(F_i(\bm Z_i;\bm\theta_i)-r_i^\star(\bm Z_i))^2\right].
\label{eq_generalization_mse_i}
\end{equation}
{Training and validation determine $F_i$. Conditional ONN optimization in \mbox{Section~\ref{sec_onn_instance}} explains fitting, and posterior-MSE bounds quantify reference decision risk. The first-error identity supports two direct certification paths. Learned margins and mean-squared output-change bounds, termed budgets, certify ONN weight updates. Joint frame errors certify fixed graph implementation combinations. Conditional on the trained functions, each path has its own moment and independent-sampling conditions.}

\vspace{-6pt}\section{Regression and Sequential Reliability}\vspace{-3pt}
\label{sec_theory}
For a fixed predictor, let $E_i=E_{\mathcal D_i}(F_i)$, and let $P_{{\rm e},i}^{\rm NN}$ and $P_{{\rm e},i}^{\rm MAP}$ denote its NN and MAP bit-error probabilities on the same law. The following bounds instantiate standard surrogate-risk principles~\mbox{\cite{bartlett2006risk}} for the signed posterior score used by the polar receiver.

\vspace{-9pt}\subsection{Direct and Posterior-Margin Bounds}\vspace{-3pt}
Let $\mathcal M_i=\{\mathbb I(F_i\ge0)\ne g_i^\star\}$ denote decision disagreement and $\alpha_i(\rho)=\mathbb P(\gamma_i\le\rho)$ the low-margin probability. Standard regression-to-decision bounds give, for $\rho\in(0,1)$,
\begin{equation}
P_{{\rm e},i}^{\rm NN}-P_{{\rm e},i}^{\rm MAP}
\le\min\{\sqrt{E_i},\rho\alpha_i(\rho)+E_i/\rho\}.
\label{eq_direct_excess_ber}
\end{equation}
The left side is excess BER relative to MAP. Its conditional value is $\gamma_i\mathbb I(\mathcal M_i)$, and on disagreement $|F_i-r_i^\star|\ge\gamma_i$. Cauchy--Schwarz yields the first bound. Splitting at $\rho$ bounds the low-margin part by $\rho\alpha_i(\rho)$ and the remainder by $\mathbb E[\gamma_i^2\mathbb I(\mathcal M_i)]/\rho\le E_i/\rho$. Weighting ambiguous decisions improves the unweighted bound $\alpha_i(\rho)+E_i/\rho^2$. If $\alpha_i(\rho)\le C_i\rho^{\kappa_i}$ for constants $C_i>0$ and $\kappa_i\ge0$, balancing the two terms gives order $E_i^{(\kappa_i+1)/(\kappa_i+2)}$ as $E_i\to0$~\mbox{\cite{audibert2007}}.

Combining these bounds for $U\ge E_i$ gives
\begin{equation}
\varGamma_i(U)=\min\!\left\{\sqrt U,\ \inf_{0<\rho<1}\left[\rho\alpha_i(\rho)+U/\rho\right]\right\}.
\label{eq_combined_ber_certificate}
\end{equation}
In particular, $U\le\varepsilon^2$ suffices for excess BER at most $\varepsilon$.

\begin{proposition}[Conversion using the margin distribution]
\label{prop_margin_distribution}
Let $(x)_+=\max\{x,0\}$. For an MSE upper bound $U$, define
\begin{equation}
\varPsi_i(U)=\inf_{\lambda\ge0}\left\{\lambda U+\mathbb E[(\gamma_i-\lambda\gamma_i^2)_+]\right\}.
\label{eq_margin_profile}
\end{equation}
Then $P_{{\rm e},i}^{\rm NN}-P_{{\rm e},i}^{\rm MAP}\le\varPsi_i(U)\le\varGamma_i(U)$. Moreover, $\varPsi_i(U)$ is the largest value of $\mathbb E[\gamma_i a(\gamma_i)]$ over measurable $a:[0,1]\to[0,1]$ subject to $\mathbb E[\gamma_i^2a(\gamma_i)]\le U$.
\end{proposition}
\begin{proof}
For every $\lambda\ge0$ and disagreement event $\mathcal M_i$,
\begin{equation}
\gamma_i\mathbb I(\mathcal M_i)\le\lambda(F_i-r_i^\star)^2+(\gamma_i-\lambda\gamma_i^2)_+.
\end{equation}
Taking expectations proves the first inequality. For $U>0$, choosing $\lambda=1/(2\sqrt U)$ and using $(x-\lambda x^2)_+\le1/(4\lambda)$ gives $\varPsi_i(U)\le\sqrt U$. Choosing $\lambda=1/\rho$ gives the stronger margin relaxation $\varPsi_i(U)\le U/\rho+(\rho/4)\alpha_i(\rho)$. Limits handle $U=0$. \mbox{Appendix~\ref{app_margin_profile}} proves the optimization characterization.
\end{proof}
The sharp relaxation assigns disagreement to the smallest positive margins. Joint confidence allocation across a fixed multiplier grid and the MSE bound permits certified minimization.

\vspace{-9pt}\subsection{Held-Out Population Certification}\vspace{-3pt}
Subsequent certificates use $\bar F_i=\operatorname{clip}(F_i,-1,1)$, which preserves hard decisions, and $E_i$ denotes its posterior MSE. For $M_i\ge2$ independent and identically distributed~\mbox{(i.i.d.)} held-out pairs, independent of training and model selection, define
the loss $L_{i,m}^{\rm te}=(\bar F_i(\bm Z_{i,m}^{\rm te})-V_{i,m}^{\rm te})^2\in[0,4]$, its sample mean $\widehat{R}_i^{\rm te}$, and unbiased sample variance $\widehat{V}_i^{\rm te}$.
The failure probability is $\delta_i\in(0,1)$. To exploit the observed loss variance, the empirical Bernstein inequality~\mbox{\cite{MaurerPontil2009}} bounds the expected squared error relative to $V_i$, termed signed-label risk, with confidence $1-\delta_i$ by
\begin{equation}
\begin{aligned}
\overline R_i^{\rm EB}(\delta_i)=\min\!\Bigg\{4,&\widehat{R}_i^{\rm te}+\sqrt{\frac{2\widehat{V}_i^{\rm te}\log(2/\delta_i)}{M_i}}\\
&+\frac{28\log(2/\delta_i)}{3(M_i-1)}\Bigg\}.
\end{aligned}
\label{eq_heldout_empirical_bernstein}
\end{equation}
Write $\overline R_i^{\rm te}=\overline R_i^{\rm EB}(\delta_i)$ and require $\sum_i\delta_i\le\delta$. Conditioning gives $\mathbb E[(\bar F_i-V_i)^2\mid\bm Z_i]=(\bar F_i-r_i^\star)^2+\operatorname{Var}(V_i\mid\bm Z_i)$, so this also bounds posterior MSE. \mbox{Appendix~\ref{app_empirical_bernstein}} derives the scaling. Signed-label risk includes conditional label variance.

Averaging the bitwise bounds gives excess reference BER at most $\varepsilon_{\rm BER}$ whenever $K^{-1}\sum_i\sqrt{\overline R_i^{\rm te}}\le\varepsilon_{\rm BER}$. Certified margins can replace each square root by~\eqref{eq_combined_ber_certificate} or~\eqref{eq_margin_profile}. Signed-label risk also bounds absolute error by $P_{{\rm e},i}^{\rm NN}\le\min\{1,\overline R_i^{\rm te}\}$, since a wrong decision has $(\bar F_i-V_i)^2\ge1$. \mbox{Section~\ref{sec_update_reliability}} uses mean-squared output-change budgets to certify subsequent predictor updates.

\vspace{-9pt}\subsection{From Bitwise Bounds to Sequential Errors}\vspace{-3pt}
\label{subsec_ga_alpha}
When posteriors are computable, MSE and margin probabilities receive simultaneous confidence bounds, with failure probability allocated across the inspected thresholds. Let $b_i=\min\{1,P_{{\rm e},i}^{\rm MAP}+\varGamma_i(\overline R_i^{\rm te})\}$ use those simultaneous bounds. Reference BER is at most $K^{-1}\sum_i b_i$. Bayes optimality permits replacing MAP by SC reference BER. Alternatively,~\eqref{eq_bounded_kl_inversion} directly assesses excess loss $\gamma_i\mathbb I(\hat u_i\ne g_i^\star)\in[0,1]$.
Reference and sequential evaluation use identical deterministic local functions, frozen bits, clipping and zero-tie rules, with no path state beyond the prefix. The bank and operating law remain fixed during use. Let $d_i^{\rm NN,ref}$ be predictor $i$\textquotesingle s decision with the transmitted prefix, and set $\mathcal E_i^{\rm ref}=\{d_i^{\rm NN,ref}\ne U_i\}$. Its block-error event has the exact representation
\begin{equation}
\{\hat{\bm m}^{\rm seq}\ne\bm m\}=\bigcup_{i\in\mathcal A_K}\mathcal E_i^{\rm ref}.
\label{eq_exact_reference_block_event}
\end{equation}
The first sequential error occurs on the true prefix. Conversely, induction shows that sequential decoding reaches the earliest reference error on that prefix. Their block-error events therefore coincide on every frame. Reference bit-error bounds $b_i^{\rm ref}$ imply
\begin{equation}
P_{\rm BER}^{\rm seq}\le P_{\rm BLER}^{\rm seq}
\le\min\!\left\{1,\sum_{i\in\mathcal A_K}b_i^{\rm ref}\right\}.
\label{eq_sequential_union}
\end{equation}
The joint frame event retains bit dependence and avoids summing marginal error bounds.

\vspace{-9pt}\subsection{Direct Assessment of Frame Events}\vspace{-3pt}
\label{sec_certification_logic}
For a fixed reference decision $d_i^{\rm R}$, excess reference BER is the neural-minus-reference error probability under transmitted-prefix evaluation. The harmful-disagreement indicator $J_i=\mathbb I(\hat u_i\ne U_i,d_i^{\rm R}=U_i)$ upper-bounds the {difference of error indicators}. Its expectation therefore directly bounds excess reference BER on the same operating law. For any frame average $X_m\in[0,1]$, define the Bernoulli Kullback--Leibler~\mbox{(KL)} upper bound $U_{\rm KL}(\bar X,M,\delta)$ as the largest $u\in[\bar X,1]$ satisfying
\begin{equation}
\operatorname{KL}(\bar X\|u)\le\log(1/\delta)/M,
\label{eq_bounded_kl_inversion}
\end{equation}
where $\operatorname{KL}(a\|b)=a\log(a/b)+(1-a)\log[(1-a)/(1-b)]$. The population mean of $X_m$ is then bounded with confidence $1-\delta$. The proof for bounded, possibly nonbinary frame averages appears in \mbox{Appendix~\ref{app_bounded_kl}}. This bounded-mean form accommodates nonbinary update envelopes and preserves within-frame bit dependence.
For two deterministic successive receivers, let $D_m$ indicate either a transmitted-prefix disagreement or that exactly one receiver fails. Either event bounds the absolute BLER gap. For $M$ independent frames and $k_D=\sum_mD_m$, the Clopper--Pearson~\mbox{(CP)} bound~\mbox{\cite{clopper1934}} is
\begin{equation}
U_D=\begin{cases}
\operatorname{Beta}^{-1}(1-\delta;k_D+1,M-k_D),&k_D<M,\\
1,&k_D=M.
\end{cases}
\label{eq_sequential_sc_exact_certificate}
\end{equation}
Here $\operatorname{Beta}^{-1}$ is a beta quantile. Write $U_{\rm CP}(k;\delta)$ for this endpoint with count $k$ and frame total $M$. With confidence $1-\delta$, the absolute BLER gap is at most $U_D$. For $k_D=0$, $U_D=1-\delta^{1/M}$, so $M\ge\log(\delta)/\log(1-\varepsilon)$ suffices for tolerance $\varepsilon$.

\vspace{-9pt}\subsection{Reliability Under Predictor Updates}\vspace{-3pt}
\label{sec_update_reliability}
Weight changes, such as quantization, perturb the bitwise decision outputs. Let $f_i$ be the clipped output of a fixed reference predictor, called the center, and $g_i$ the updated output. Both use transmitted preceding bits and the same operating law. For the updated decoder $g$, excess BLER is $P_{\rm BLER}^{g}-P_{\rm BLER}^{\rm SC}$, where the two terms are its sequential block-error probability and that of SC. The center decision margin $|f_i|$ is its output distance to the zero threshold. Define
\begin{equation}
\begin{gathered}A=\mathbb I(\exists i:\mathbb I(f_i\ge0)\ne U_i),\\
S=\mathbb I(\text{SC block error}),\qquad m=\min_i|f_i|.\end{gathered}
\label{eq_update_center}
\end{equation}
The first-error identity makes $A$ the center's sequential block-error indicator. Write $C=(1-A)(1-S)$ and $J=A(1-S)$. They identify joint success and center-only failure, respectively. A budget $\varepsilon_i\ge0$ is an upper bound on $\mathbb E(g_i-f_i)^2$, the mean-squared output change at position $i$. The update class consists of all deterministic predictor banks satisfying these constraints.

\begin{proposition}[Position-resolved update reliability]
\label{prop_update_envelope}
If $\mathbb E(g_i-f_i)^2\le\varepsilon_i$ for each information position, every deterministic updated bank satisfies
\begin{equation}
\begin{aligned}
P_{\rm BLER}^{g}-P_{\rm BLER}^{\rm SC}
&\le\varPhi(\bm\varepsilon)
=\inf_{\bm\lambda\ge\bm0}\{\bm\lambda^\top\bm\varepsilon+\mathbb E Q_{\bm\lambda}\},\\
Q_{\bm\lambda}&=J+C\bigl(1-\min_i\lambda_i f_i^2\bigr)_+.
\end{aligned}
\label{eq_position_update}
\end{equation}
Moreover, this upper bound equals the following optimization over relaxed first-error indicators, subject to the same expected squared-output constraints:
\begin{equation}
\varPhi(\bm\varepsilon)=\mathbb E J+
\sup_{\substack{a_i\ge0,\ \sum_i a_i\le1\\
\mathbb E[C f_i^2 a_i]\le\varepsilon_i\ \forall i}}
\mathbb E\left[C\sum_i a_i\right],
\label{eq_position_primal}
\end{equation}
Here $\bm\lambda=(\lambda_i)_{i\in\mathcal A_K}$ contains nonnegative Lagrange multipliers, and each measurable $a_i$ represents a relaxed first-error indicator at position $i$. Allowing fractional values subject to the stated moment constraints enlarges the set of realizable network-error patterns. This optimization is termed a moment relaxation.
\end{proposition}
\begin{proof}
Let $T_i$ indicate that $C=1$ and position $i$ is the updated bank's first reference error. These events are disjoint, and crossing zero requires $f_i^2T_i\le(g_i-f_i)^2$. With $r=\min_i\lambda_i f_i^2$, their union indicator satisfies
$\sum_iT_i\le\sum_i\lambda_i(g_i-f_i)^2+C(1-r)_+$.
Center-only errors contribute at most $J$. Taking expectations proves the bound. The first-error indicators are feasible allocations in~\eqref{eq_position_primal}. Strong duality, including zero budgets, is proved in \mbox{Appendix~\ref{app_position_budget}}.
\end{proof}

Restricting $\lambda_i=\lambda$ recovers the total-budget envelope
\begin{equation}
\varPhi(\bm\varepsilon)\le
\inf_{\lambda\ge0}\{\lambda\varepsilon+\mathbb E[J+C(1-\lambda m^2)_+]\},
\qquad\varepsilon=\sum_i\varepsilon_i.
\label{eq_update_mse_envelope}
\end{equation}
The total-budget relaxation imposes the aggregate constraint $\sum_i\mathbb E(g_i-f_i)^2\le\varepsilon$. Position-specific constraints additionally preserve the allocation of perturbations across information positions and yield a potentially tighter population bound. The gain can be strict: for $C=1$, squared margins $(0.01,1)$ and budgets $(0,0.1)$, the position-resolved relaxation is $0.1$, whereas the total-budget relaxation is one. Separate budgets prevent a perturbation allowance at a reliable position from being reassigned to a more vulnerable position. Sharpness refers to the moment relaxation, whose allocations may exceed those realizable by a specified causal network.

If $\max_i|g_i-f_i|\le r$ almost surely, the same crossing argument gives
\begin{equation}
P_{\rm BLER}^{g}-P_{\rm BLER}^{\rm SC}
\le\mathbb E[J+C\mathbb I(m\le r)].
\label{eq_update_uniform_envelope}
\end{equation}
For a finite set $\mathcal T$ of nonnegative multiplier vectors, fixed independently of the assessment frames, $Q_{\bm\lambda}\in[0,1]$. Simultaneous KL bounds $U_{\bm\lambda}$ therefore give
$\min_{\bm\lambda\in\mathcal T}\{\bm\lambda^\top\bm\varepsilon+U_{\bm\lambda}\}$
for every budget vector on one event. Multiplier design may use a separate sample. Finite-sample tightness also depends on multiplier selection and confidence allocation. A population budget estimated from data requires its own valid upper bound.

For ONNs with fixed output signs and normalization $B^{-1/2}$, Cauchy--Schwarz and nonexpansiveness of ReLU and clipping yield $|g_i-f_i|^2\le\|\Delta\mathbf W_i^\top\bm Z_i\|_2^2$. The input second moment $\mathbf M_i=\mathbb E[\bm Z_i\bm Z_i^\top]$ supplies each budget:
\begin{equation}
\begin{aligned}
\mathbb E(g_i-f_i)^2\le\varepsilon_i
&=\operatorname{tr}(\Delta\mathbf W_i^\top\mathbf M_i\Delta\mathbf W_i)\\
&\le\|\Delta\mathbf W_i\|_{\rm F}^2\operatorname{tr}(\mathbf M_i).
\end{aligned}
\label{eq_update_parameter_budget}
\end{equation}
These quantities retain perturbation directions and observation-prefix correlations. Known moments give deterministic budgets for the real functions defined by stored weights, independently of fixed-data optimization.

\vspace{-9pt}\subsection{Joint Certification of Local Model Choices}\vspace{-3pt}
\label{sec_precision_assemblies}
At information position $i$, a local library $\mathcal G_i$ is a finite implementation set containing center $f_i$, fixed before assessment. Experimental versions quantize the trained predictor or simplify its readout without retraining. One version per position defines a decoder, termed an assembly, in $\prod_i\mathcal G_i$, fixed during operation. Assessment reuses transmitted-prefix outputs across assemblies. The result applies to any fixed sets satisfying its conditions.

\begin{proposition}[A common-frame certificate for model assemblies]
\label{prop_precision_assemblies}
For fixed local implementation sets and a common transmitted-prefix frame, define
\begin{equation}
\begin{gathered}
d_i=\max_{g\in\mathcal G_i}|g(\bm Z_i)-f_i(\bm Z_i)|,\\
Q_{\rm rad}=J+C\mathbb I(\exists i:|f_i|\le d_i),\\
Q_{\rm set}=(1-S)\mathbb I(\exists i,g\in\mathcal G_i:
\mathbb I(g(\bm Z_i)\ge0)\ne U_i).
\end{gathered}
\label{eq_precision_events}
\end{equation}
Here $d_i$ is the largest candidate output deviation at position $i$. Every assembly, including one selected after certification, satisfies
\begin{equation}
P_{\rm BLER}^{g}-P_{\rm BLER}^{\rm SC}
\le\mathbb E[Q_{\rm set}]\le\mathbb E[Q_{\rm rad}].
\label{eq_precision_assemblies}
\end{equation}
The event $Q_{\rm rad}$ uses these deviations to bound crossings of the reference decision thresholds. It also covers every predictor bank with $|g_i-f_i|\le d_i$ almost surely. Each independent-frame binomial bound therefore certifies the entire class at a common confidence level.
\end{proposition}
\begin{proof}
An assembly failure with SC success entails a local transmitted-prefix error, hence $Q_{\rm set}=1$. If the center succeeds, this error crosses its threshold, so $|f_i|\le d_i$ somewhere. Otherwise $J=1$. Thus $Q_{\rm set}\le Q_{\rm rad}$ pointwise. The crossing argument also covers arbitrary scores within the stated deviations. Both events depend only on the fixed sets and frame, so their common confidence event covers subsequent assembly choices.
\end{proof}

Frames on which SC fails and every candidate decoder succeeds contribute a common performance gain. Define
\begin{equation}
R_{\rm set}=S\mathbb I(\forall i,g\in\mathcal G_i:\mathbb I(g(\bm Z_i)\ge0)=U_i).
\end{equation}
Every assembly is correct on this event, so
\begin{equation}
P_{\rm BLER}^{g}-P_{\rm BLER}^{\rm SC}\le\mathbb E[Q_{\rm set}-R_{\rm set}].
\label{eq_signed_assembly}
\end{equation}
The shared loss $Q_{\rm set}-R_{\rm set}$ is the largest neural-minus-SC block-error indicator difference on each frame. A local error makes some assembly fail there or earlier. With no local errors, all assemblies succeed. The maximizing assembly may vary across frames. The paired binomial certificate subtracts a binomial lower bound for $R_{\rm set}$ from an upper bound for $Q_{\rm set}$, each at $\tfrac{\delta}{2}$, to give confidence $1-\delta$. Empirical Bernstein concentration provides a variance-sensitive alternative for {their $[-1,1]$ difference}. For a family containing one receiver, this reduces to direct bank assessment.

The candidate-set event covers evaluated predictors. Deterministic output-deviation bounds, including those from quantization verification~\mbox{\cite{zhang2023qebverif}}, allow the radius event to cover unevaluated predictors as well.

\vspace{-9pt}\subsection{Gap between the Shared Bound and the Worst Fixed Decoder}\vspace{-3pt}
\label{sec_assembly_gap}
Let $\mathcal G=\prod_i\mathcal G_i$, let $\mathcal B_g$ be assembly $g$'s block-error event, and set $\mathcal B_\cup=\bigcup_g\mathcal B_g$. By the first-error identity, $\mathcal B_\cup$ occurs exactly when at least one local implementation errs on the transmitted prefix.
\begin{proposition}[Exact common-event gap]
\label{prop_assembly_gap}
For a finite fixed family,
\begin{equation}
\begin{split}
&\mathbb E[Q_{\rm set}-R_{\rm set}]
-\max_{g\in\mathcal G}(P_{\rm BLER}^{g}-P_{\rm BLER}^{\rm SC})\\
&\hspace{12mm}=\min_{g\in\mathcal G}\mathbb P(\mathcal B_\cup\setminus \mathcal B_g).
\end{split}
\label{eq_assembly_gap}
\end{equation}
The gap quantifies the effect of maximizing separately on each frame instead of using one fixed decoder. Statistical certificates also include finite-sample corrections. The gap vanishes exactly when one assembly contains all assembly-error events almost surely. If $L=|\mathcal G|$ and $p_*=\max_g\mathbb P(\mathcal B_g)$, the gap is at most $\min\{1,Lp_*\}-p_*$.
\end{proposition}
\begin{proof}
Pointwise, $Q_{\rm set}-R_{\rm set}=\mathbb I(\mathcal B_\cup)-S$. Subtracting the largest fixed-assembly expectation cancels the reference term and gives~\eqref{eq_assembly_gap}. The final bound follows from $\mathbb P(\mathcal B_\cup)\le\min\{1,\sum_g\mathbb P(\mathcal B_g)\}$.
\end{proof}
For each fixed $g$, $\mathbb I(\mathcal B_\cup\setminus\mathcal B_g)$ is Bernoulli. Simultaneous binomial intervals $[l_g,u_g]$ therefore place the population gap in $[\min_g l_g,\min_g u_g]$ by~\eqref{eq_assembly_gap}. This simultaneous coverage also permits selection of the empirically worst receiver.
If only one position errs, with $L$ choices on disjoint equiprobable events, fixed-assembly BLER is $1/L$ and the common-event probability is one. The gap attains $1-1/L$. Nested events give zero gap. \mbox{Section~\ref{sec_enumeration}} examines actual ONN choices.

\vspace{-9pt}\subsection{A Computational Tradeoff in Family Certification}\vspace{-3pt}
\label{sec_computational_tradeoff}
Let $U_g(\delta)$ be a valid $1-\delta$ upper bound for the fixed bank's excess $\varDelta_g=P_{\rm BLER}^g-P_{\rm BLER}^{\rm SC}$. A scalar family bound is
\begin{equation}
\mathbb P\!\left\{\max_{g\in\mathcal G}\varDelta_g\le\max_{g\in\mathcal G}U_g(\delta)\right\}\ge1-\delta.
\label{eq_maximum_family_bound}
\end{equation}
A population maximizer $g_\star$ is fixed independently of assessment. Its event $\varDelta_{g_\star}\le U_{g_\star}$ suffices to cover the family at one member\textquotesingle s confidence level. Simultaneous coverage of individually reported member bounds requires a separate confidence allocation. For paired binomial bounds at equal $\delta$, each member has no more harmful events and no fewer gains than the common event. Monotonicity therefore gives $\max_gU_g\le U_{\rm common}$. The common construction reduces computation at the cost of a potentially looser bound. Its sample requirements depend on the overlap of local errors.

\begin{proposition}[Worst fixed-assembly assessment]
\label{prop_assembly_complexity}
Given arbitrary finite local error tables with two choices at each position, computing the largest empirical block-error rate of a fixed assembly is NP-hard. The common event is computable in $\mathcal O(M\sum_i|\mathcal G_i|)$ operations.
\end{proposition}
\begin{proof}
Reduce Boolean satisfiability with three literals per clause~\mbox{\cite{garey1979}} to this problem. Associate positions with variables and frames with clauses. Set $e_{mi,b}=1$ exactly when assigning variable $i$ value $b$ satisfies a literal in clause $m$. A fixed assembly $\bm b$ then has frame error $\max_i e_{mi,b_i}$, so its empirical block-error rate is the fraction of satisfied clauses. Its maximum equals one exactly when all clauses are satisfiable. These tables can be realized by deterministic lookup predictors on distinct observations with a fixed transmitted message. This NP-hardness result applies to unrestricted local tables. The complexity of a particular trained polar library depends on its error structure. The common event is computed as $\max_{i,b}e_{mi,b}$ on each frame through a single scan of the table.
\end{proof}
Exhaustive enumeration evaluates the bound for every fixed decoder combination and takes their maximum. After shared predictor inference, this costs up to $\mathcal O(MK|\mathcal G|)$. Fixed-sequence testing orders candidates with independent design data, tests at level $\delta$ and stops at the first failure. A false acceptance requires rejecting the first true null, which gives familywise control. All members pass exactly when the maximum-bound test passes.

\vspace{-9pt}\subsection{A Verified Relaxation for Family Assessment}\vspace{-3pt}
\label{sec_verified_relaxation}
Each candidate decoder uses one predictor per position on every frame. An LP relaxation bounds the largest error count subject to these consistent choices. Let $e_{mib}\in\{0,1\}$ indicate a transmitted-prefix decision error. The indices $m\in\mathcal F$, $i\in\mathcal A_K$, and $b$ specify the assessment frame, information position, and candidate implementation, respectively. With $x_{ib}$ selecting option $b$ at position $i$ and $y_m$ representing a frame-error indicator, the maximum error count over fixed decoder combinations is the optimum of
\begin{equation}
\begin{gathered}\max\sum_{m\in\mathcal F}y_m,\qquad
y_m\le\sum_{i,b}e_{mib}x_{ib},\\ 0\le y_m\le1,\qquad
\sum_bx_{ib}=1.\end{gathered}
\label{eq_count_lp}
\end{equation}
Relaxing $x_{ib}\in\{0,1\}$ to $x_{ib}\ge0$ gives an upper bound. This standard relaxation admits an independently verifiable certificate: for any $\alpha_m\in[0,1]$,
\begin{equation}
D_{\mathcal F}(\bm\alpha)=
\sum_{m\in\mathcal F}(1-\alpha_m)+
\sum_i\max_b\sum_{m\in\mathcal F}\alpha_m e_{mib}.
\label{eq_count_dual}
\end{equation}
Indeed, $y_m\le1-\alpha_m+\alpha_m\sum_{i,b}e_{mib}x_{ib}$, and maximizing each simplex yields~\eqref{eq_count_dual}. The solver's dual vector is quantized to multiples of $2^{-20}$ in $[0,1]$, and the numerator is evaluated with exact integer arithmetic. The smaller of $\lfloor D_{\mathcal F}\rfloor$ and the common-union count is a valid upper count regardless of solver tolerances. Identical rows of multiplicity $w$ use an aggregate dual in $[0,w]$. Frames on which all candidate decoders fail or all succeed contribute one and zero, respectively.

\begin{proposition}[Verified paired-family bound]
\label{prop_verified_relaxation}
Let $\bar h$ upper-bound the largest bank error count on SC-correct frames, and let $\bar b$ do so on the $n_S$ SC-error frames, using~\eqref{eq_count_dual}. Then, with probability at least $1-\delta$, every bank satisfies
\begin{equation}
\varDelta_g\le U_{\rm LP}
=U_{\rm CP}(\bar h;\tfrac{\delta}{2})-L_{\rm CP}(n_S-\bar b;\tfrac{\delta}{2}).
\label{eq_verified_lp_cp}
\end{equation}
Both CP endpoints use the original frame total $M$, and $L_{\rm CP}(g;a)=\operatorname{Beta}^{-1}(a;g,M-g+1)$, with $L_{\rm CP}(0;a)=0$.
\end{proposition}
\begin{proof}
For each assembly, let $h_g$ count neural failures on SC-correct frames and $g_g$ count neural successes on SC-error frames. Then $h_g\le\bar h$ and $g_g\ge n_S-\bar b$. Monotonicity of the paired binomial bound gives $U_{\rm LP}\ge\max_gU_g(\delta)$ for every sample. \mbox{Equation~\eqref{eq_maximum_family_bound}} then proves the claim.
\end{proof}
The dual variables may be optimized using the assessment outputs. The paired bound combines separate maximizations of the harmful and SC-error counts, each over its own choice of predictor bank.

\vspace{-9pt}\subsection{Adaptive Refinement at a Fixed Confidence Level}\vspace{-3pt}
\label{sec_partition_certificate}
Dividing the candidate decoder set into smaller subsets gives tighter common-loss and LP bounds. Every refinement upper-bounds the same maximum of paired binomial bounds. Refinement depth can therefore be selected from the assessment data at the same confidence level. 
\begin{proposition}[Adaptive family envelopes]
\label{prop_partition_certificate}
For any cover $\mathcal P$, a collection of subsets whose union is the fixed family, let $V_{\mathcal H}$ dominate $\max_{g\in\mathcal H}U_g(\delta)$ on every assessment sample. Covers, dual iterates and stopping decisions may depend on that sample. Then
\begin{equation}
\max_g\varDelta_g\le
\inf_{\mathcal P\ \mathrm{assessed}}\max_{\mathcal H\in\mathcal P}V_{\mathcal H}
\label{eq_partition_bound}
\end{equation}
holds with probability at least $1-\delta$. Common and verified LP bounds at the same $\delta$ are admissible bounds for each subset.
\end{proposition}
\begin{proof}
Every cover maximum dominates $\max_gU_g(\delta)$ pointwise, as does their infimum. The single population-worst-member event in~\eqref{eq_maximum_family_bound} therefore suffices, irrespective of how covers or duals are selected.
\end{proof}
Fixing the first $d$ binary choices gives $2^d$ subfamilies. The common error events for refined subsets contain no more harmful events and no fewer gains, so their paired bounds decrease under refinement at equal $\delta$. Depth zero gives the common bound and full depth the exhaustive maximum. Computation can stop at any intermediate depth while retaining an upper bound on the same $U_g$. Minimization across separately calibrated methods requires a corresponding joint confidence allocation.

\vspace{-9pt}\subsection{Selecting an Implementation after Certification}\vspace{-3pt}
\label{sec_selection_protocol}
For libraries $\mathcal G_i$ certified under the specified operating laws, storage or replacement constraints define feasible decoders $\mathcal H\subseteq\prod_i\mathcal G_i$. The joint guarantee permits assessment-dependent selection of a fixed decoder.
Simultaneous bounds $U_1,U_2$ imply $wU_1+(1-w)U_2$ for every mixture weight $w\in[0,1]$, chosen after assessment and fixed before operation. New predictors or operating laws require a new assessment.

Direct testing covers a predetermined decoder. Bonferroni jointly protects individual bounds. Fixed-sequence testing follows an independent order until failure. Exhaustive, shared-loss and LP bounds cover the family at one member\textquotesingle s confidence level. \mbox{Section~\ref{sec_budget_selection}} compares their costs.

\vspace{-6pt}\section{Fixed-Sample ONN Optimization}\vspace{-3pt}
\label{sec_onn_instance}
A one-hidden-layer ReLU ONN supports conditional MSE convergence and, through nonexpansiveness and fixed output signs, the weight-update budgets in~\eqref{eq_update_parameter_budget}. At information position $i$, it is
\begin{equation}
F_i(\mathbf W_i,\bm a_i,\bm z_i)
=\frac1{\sqrt B}\sum_{j=1}^B a_{i,j}\sigma(\bm w_{i,j}^{\top}\bm z_i),
\label{eq_three_layer_model}
\end{equation}
Here $B$ is the hidden-layer width, $D$ is the number of training pairs, and $\sigma(x)=\max\{x,0\}$. Further, $\mathbf W_i=[\bm w_{i,1},\ldots,\bm w_{i,B}]$, and $\bm a_i$ has fixed independent uniform signs. Gaussian-initialized hidden weights are trained by full-batch GD on $\ell_{\mathrm{tr},i}=(2D)^{-1}\sum_\ell(F_i(\bm z_{i,\ell})-t_{i,\ell})^2$, with $|t_{i,\ell}|\le1$. Define
\begin{equation}
[\mathbf H_i^\infty]_{\ell,\ell'}=
\bm z_{i,\ell}^{\top}\bm z_{i,\ell'}\,
\mathbb E_{\bm w\sim\mathcal N(\bm0,\mathbf I)}
\mathbb I(\bm w^\top\bm z_{i,\ell}\ge0,\bm w^\top\bm z_{i,\ell'}\ge0).
\label{eq_H_infty_i}
\end{equation}
For bounded input norms and $\lambda_{0,i}=\lambda_{\min}(\mathbf H_i^\infty)>0$, established theory~\mbox{\cite[Theorem~4.1 and Corollary~4.1]{du2018gradient}} gives geometric empirical-MSE contraction when $B\ge C_BD^6/(\lambda_{0,i}^4\delta_i^3)$ and $0<\eta_i\le C_\eta\lambda_{0,i}/D$, with probability $1-\delta_i$ over initialization. Constants depend only on the input-norm bound. For the mean-loss normalization above, the corresponding sum-loss step is $\eta_i/D$. The contraction factor is $1-\eta_i\lambda_{0,i}/(2D)$, and neuron displacement is of order $B^{-1/2}$ for fixed data and confidence. The empirical widths below are smaller than this sufficient condition. ONN weight budgets and independent population certificates apply under their respective moment and sampling conditions.

\vspace{-6pt}\section{Numerical Results}\vspace{-3pt}
\label{sec_experiments}

Experiments use the $(8,4)$ and $(128,64)$ codes to assess short-code ONN learning and reliability certification of longer-code graph decoders, respectively. ONNs use full-batch GD, and graph predictors use adaptive moment estimation~\mbox{(Adam)} with fresh samples. Assessment uses uniform messages and independent noise, with coverage in \mbox{Table~\ref{tab_protocol_index}}. {Cached statistics and LP assessment run on the central processing unit~\mbox{(CPU)}, with HiGHS as the LP solver. Neural, SC and SCL decoding run on an RTX 4090 graphics processing unit~\mbox{(GPU)}.}

\begin{table*}[t]\centering
\caption{Code configurations, samples and joint confidence guarantees for the principal certification studies.}
\label{tab_protocol_index}
\setlength{\tabcolsep}{3pt}
\begin{tabular}{@{}>{\raggedright\arraybackslash}m{.31\textwidth}>{\centering\arraybackslash}m{.12\textwidth}>{\centering\arraybackslash}m{.14\textwidth}>{\raggedright\arraybackslash}m{.39\textwidth}@{}}\toprule
 Code and study & Frames per law & Joint confidence & Coverage of the guarantee\\\midrule
$(8,4)$ posterior-risk conversion & $2\!\times\!10^{5}$ & $\ge95\%$ & Selected decoder: risk and margin-based bounds\\
$(8,4)$ weight updates & $2\!\times\!10^{5}$ & $\ge95\%$ & Five update classes: total and position-specific budgets\\
$(8,4)$ predictor combinations & $2\!\times\!10^{5}$ & $\ge95\%$ per method & All decoder combinations in four candidate families\\
$(128,64)$ readout versions & $5\!\times\!10^{4}$ & $\ge95\%$ & Selected decoders and all combinations of four readout versions\\
$(128,64)$ graph precisions & $5\!\times\!10^{4}$ & $\ge95\%$ & Four fixed-precision decoders and three candidate families\\
$(128,64)$ and $(64,48)$ mixed-precision families & $5\!\times\!10^{4}$ & $\ge95\%$ per code & Baseline and mixed families, plus selected decoder\\
$(128,64)$ adaptive partition refinement & $5\!\times\!10^{4}$ & $\ge95\%$ & All assessed depths under four channel--precision settings\\
$(128,64)$ and $(64,48)$ family and sample sizes & $5000$, $2\!\times\!10^{4}$, $5\!\times\!10^{4}$ & $\ge95\%$ per method and code & All $48$ settings for each code\\\bottomrule
\end{tabular}
\end{table*}
Certification samples are independent of training and candidate design. Each study has a separate confidence event, with joint allocation required for minimization across methods. The size comparisons in \mbox{Section~\ref{sec_assembly_scaling}} reuse mixed-family assessment samples. Partition refinement in \mbox{Section~\ref{sec_partition_experiment}} assesses its predetermined families on the selection samples of \mbox{Section~\ref{sec_budget_selection}}. Gap intervals have separate events.

\vspace{-9pt}\subsection{{Noiseless Short-Code Optimization}}\vspace{-3pt}
\label{subsec_short_code_setup_v2}

The noiseless optimization study uses all $16$ messages of the $(8,4)$ code. Inputs are $[-\bm y,\bm u^{i-1}]/\sqrt{8+i-1}$, and targets use SC with exact parity-message updates at nominal variance $1.2619147$.

The minimum eigenvalues of $\mathbf H_i^\infty$ at information positions $4,6,7,8$ are $0.3636364$, $0.2467188$, $0.2133078$, and $0.1959114$, respectively. With $D=16$, $\eta_i=0.5\lambda_{0,i}/D$ and widths $128$--$2048$, all $20$ position-wise MSE trajectories satisfy their initialization-based contraction references over $4\!\times\!10^{4}$ updates. \mbox{Figure~\ref{fig_short_train_v2}} plots four-position averages. At the final update, each predictor is characterized by its maximum neuron displacement. Log--log least-squares fits of its mean and maximum across positions against width yield slopes $-0.4859$ and $-0.4693$, with coefficients of determination $R^2>0.97$. \mbox{Figure~\ref{fig_short_drift_v2}} plots the position with the largest final displacement at each width.

\begin{figure}[t]
\centering
\includegraphics[width=2.90in]{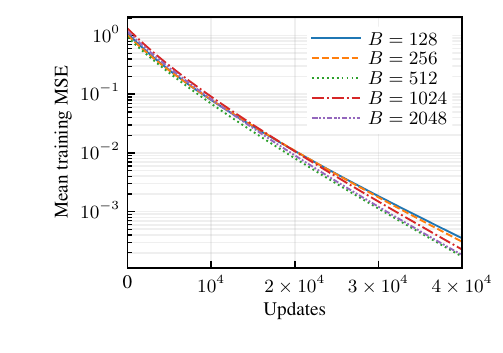}
\caption{Mean fixed-sample training MSE for five representative noiseless $(8,4)$ widths. Each curve averages the four information-bit predictors.}
\label{fig_short_train_v2}
\end{figure}

\begin{figure}[t]
\centering
\includegraphics[width=2.90in]{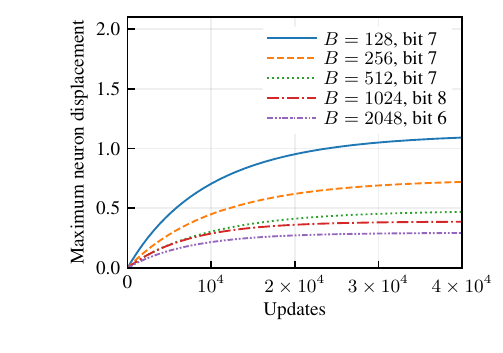}
\caption{Maximum neuron-wise displacement for five representative noiseless $(8,4)$ widths. Each curve selects the information bit with the largest final displacement at that width.}
\label{fig_short_drift_v2}
\end{figure}

\vspace{-9pt}\subsection{{Short-Code Reliability Under Noise}}\vspace{-3pt}
\label{subsec_large_width_core_v2}

{The noisy $(8,4)$ study uses $5$~dB, with $\sigma_{\rm n}^2=10^{-0.5}$ under \mbox{Section~\ref{sec_preliminaries}}. Each of $16$ messages supplies $1024$ draws, so $D=1.6384\!\times\!10^4$. Inputs are $[\bm y,\bm u^{i-1}]/\sqrt{8+i-1}$, and targets enumerate the prefix-conditional codebook posterior. Widths $32,128,512,2048,8192$ use $10^4$ full-batch GD steps at rate $8$. Minimum posterior MSE on $1.6384\!\times\!10^4$ independent validation frames selects initialization or each $50$th update. SC chooses bit zero at zero LLR, while neural and MAP scores choose one. Each rule is retained in both evaluation modes.}

\mbox{Figs.~{\hyperref[fig_large_width_summary_v2_b]{\ref*{fig_large_width_summary_v2}(b)}} and~{\hyperref[fig_large_width_summary_v2_c]{\ref*{fig_large_width_summary_v2}(c)}}} use $2\!\times\!10^{5}$ new common test frames. Posterior MSE is shown in \mbox{Fig.~{\hyperref[fig_large_width_summary_v2_a]{\ref*{fig_large_width_summary_v2}(a)}}}. At $B=2048$, reference and sequential BER are $1.0625\!\times\!10^{-3}$ and $1.6575\!\times\!10^{-3}$, with $687$ block errors. Min-sum SC and SCL with list size eight~\mbox{(SCL-8)} give sequential BER $1.565\!\times\!10^{-3}$ and $1.230\!\times\!10^{-3}$. Reference and sequential block-error events coincide on every frame at every width, as~\eqref{eq_exact_reference_block_event} predicts. Channel-only bit-MAP and block maximum-likelihood~\mbox{(ML)} references use the received signal and the code constraints.

\begin{figure*}[t]
\centering
\captionsetup[subfloat]{farskip=0pt,nearskip=0pt,captionskip=1pt}
\setlength{\abovecaptionskip}{2pt}
\subfloat[Validation posterior MSE\label{fig_large_width_summary_v2_a}]{\includegraphics[width=2.80in]{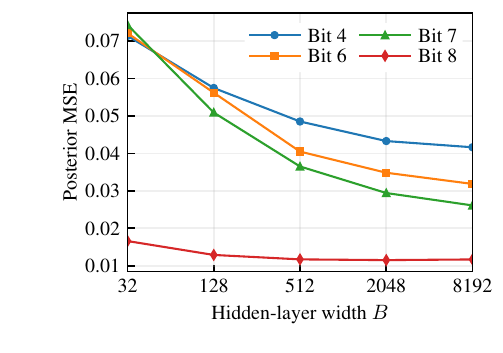}}\hspace{0.01in}
\subfloat[Bit-error rate\label{fig_large_width_summary_v2_b}]{\includegraphics[width=2.80in]{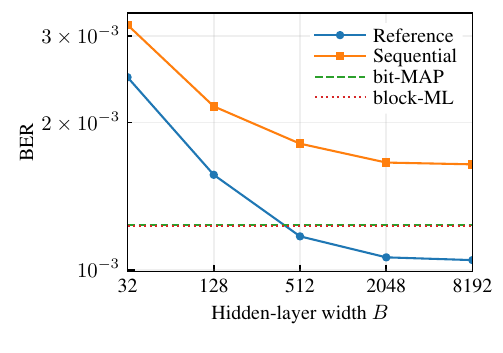}}\\[1.2ex]
\subfloat[Block-error rate\label{fig_large_width_summary_v2_c}]{\includegraphics[width=2.80in]{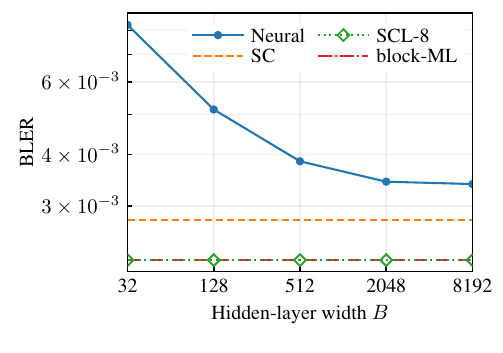}}\hspace{0.01in}
\subfloat[Posterior-risk conversion\label{fig_large_width_summary_v2_d}]{\includegraphics[width=2.80in]{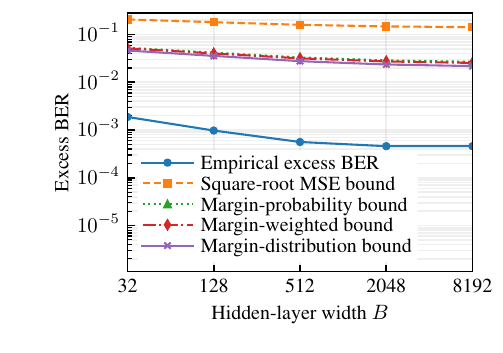}}
\caption{Short-code performance and posterior-risk conversion at $5$~dB.}
\label{fig_large_width_summary_v2}
\end{figure*}

\mbox{Fig.~{\hyperref[fig_large_width_summary_v2_d]{\ref*{fig_large_width_summary_v2}(d)}}} compares clipped-score conversions with empirical excess BER, the sample mean of posterior-weighted decision disagreements. Write $\widehat{E}_i$ for sample clipped posterior MSE and $\widehat{\alpha}_i(\rho)$ for the sample fraction with $\gamma_i\le\rho$. The margin-probability and margin-weighted bounds use $\widehat{\alpha}_i(\rho)+\widehat{E}_i/\rho^2$ and $\rho\widehat{\alpha}_i(\rho)+\widehat{E}_i/\rho$, respectively, with thresholds optimized per information position. The margin-distribution bound evaluates~\eqref{eq_margin_profile} with empirical margins and MSE. The raw and clipped posterior MSE at $B=2048$ average $3.0216\!\times\!10^{-2}$ and $2.2932\!\times\!10^{-2}$.

The designated $B=2048$ bank is certified on another $2\!\times\!10^{5}$ independently generated frames. A failure probability of $0.01$ is allocated to each of five groups: posterior MSE, signed-label risk, low-margin mass at $\rho=0.5$, multiplier terms and posterior-weighted excess. With a Bernstein bound $U_{\rm agg}$ on $K^{-1}\sum_iE_i$, the margin-distribution certificate uses one common multiplier:
\begin{equation}
\min_{\lambda\in\mathcal T}\{\lambda U_{\rm agg}+U_{X,\lambda}\},\quad
X_{\lambda,m}=\frac1K\sum_i(\gamma_{i,m}-\lambda\gamma_{i,m}^2)_+.
\label{eq_aggregate_margin_certificate}
\end{equation}
The set $\mathcal T$ contains zero and $256$ log-spaced points from $10^{-2}$ to $1000$. Each $U_{X,\lambda}$ is a Bernstein bound at $\delta\approx3.8911\!\times\!10^{-5}$, with range $[0,1-\lambda]$ for $\lambda\le1/2$ and $[0,1/(4\lambda)]$ otherwise. Frame averages retain bit dependence. Low-margin mass and direct excess use KL inversion.

The posterior-MSE upper bound of $2.3744\!\times\!10^{-2}$ implies a square-root bound of $0.15410$. The low-margin mass bound of $9.91\!\times\!10^{-4}$ at threshold $0.5$ gives an excess-BER bound of $4.7984\!\times\!10^{-2}$. \mbox{Equation~\eqref{eq_aggregate_margin_certificate}} reduces this bound to $2.4482\!\times\!10^{-2}$. Direct posterior-weighted excess has mean $4.96952\!\times\!10^{-4}$ and upper bound $6.64\!\times\!10^{-4}$. Signed-label risk bounds absolute reference BER by $2.5709\!\times\!10^{-2}$. Observed decision disagreements yield a tighter excess-risk certificate than posterior-MSE conversion in this setting.

Certified population MSE yields an error-rate guarantee. Posterior-weighted errors, learned margins and joint frame events provide additional information for tighter bounds.

\subsubsection{Position-Resolved Weight-Update Guarantees}
The four $B=2048$ predictors use $b$-bit quantized weights, neuron-wise scales $\max_j|w_j|/(2^{b-1}-1)$ and nearest-even rounding. Scales and dequantization use $64$- and $32$-bit floating point, respectively. Zero vectors and output signs stay fixed. For prefix length $j=i-1$, set $\bm c_i(\bm u)=[\bm s(\bm u),\bm u^{i-1}]$ and let $\mathcal U$ contain the $16$ legal source vectors. Then
\begin{equation}
\mathbf M_i=\frac{16^{-1}\sum_{\bm u\in\mathcal U}\bm c_i(\bm u)\bm c_i(\bm u)^\top+\sigma_{\rm n}^2\mathbf D_i}{8+j},
\label{eq_onn_input_moment}
\end{equation}
where $\mathbf D_i=\operatorname{diag}(\mathbf I_8,\mathbf0_{j\times j})$. This exact moment and the stored weight differences determine~\eqref{eq_update_parameter_budget}. The resulting envelopes cover all real-valued predictor functions satisfying the position-wise score budgets. Direct frame-event certificates assess floating-point execution.

For each allocation, $5\!\times\!10^{4}$ design frames minimize $\bm\lambda^\top\bm\varepsilon+\widehat{\mathbb E}_{\rm design}Q_{\bm\lambda}$. A bounded scalar search in $\log\lambda$ initializes Powell search over four log multipliers, retaining the scalar result if the objective increases. The selected multipliers are assessed on $2\!\times\!10^{5}$ independent frames. Failure probability $0.005$ per bounded-mean KL endpoint gives joint $95\%$ coverage of all ten endpoints. Every fixed nonnegative multiplier is valid, so numerical minimization improves tightness. Both methods use the same center, moments and assessment frames.

\mbox{Table~\ref{tab_onn_parameter_update}} compares uniform $16$-bit weights with a four-bit precision reduction at each information position separately. This isolates the effect of the changed position. The columns use combined or separate mean-squared output-change budgets.

\begin{table}[t]\centering
\caption{Upper bounds on the BLER difference from SC for ONN weight-update classes at $5$~dB.}
\label{tab_onn_parameter_update}
\small\setlength{\tabcolsep}{3pt}
\begin{tabular}{@{}ccc@{}}\toprule
\multirow{3}{*}{\shortstack{Weight precision (bits)\\Positions $(4,6,7,8)$}} & \multicolumn{2}{c}{Upper bound on BLER}\\
& \multicolumn{2}{c}{difference from SC ($10^{-3}$)}\\
& Combined budget & Position budgets\\\midrule
$(16,16,16,16)$ & $3.307$ & $3.211$\\
$(12,16,16,16)$ & $9.724$ & $7.785$\\
$(16,12,16,16)$ & $9.683$ & $7.290$\\
$(16,16,12,16)$ & $9.781$ & $6.891$\\
$(16,16,16,12)$ & $9.466$ & $4.393$\\\bottomrule
\end{tabular}
\end{table}

\begin{table}[t]\centering
\caption{Budgets, selected multipliers and bound decomposition for precision allocation $(16,16,16,12)$.}
\label{tab_position_diagnostics}
\small\setlength{\tabcolsep}{4pt}
\begin{tabular}{@{}lrrrr@{}}\toprule
Position $i$ & $4$ & $6$ & $7$ & $8$\\\midrule
$\varepsilon_i$ ($10^{-7}$) & $5.711$ & $5.747$ & $5.864$ & $1371$\\
$\lambda_i$ & $245.747$ & $242.809$ & $252.198$ & $5.514$\\\bottomrule
\end{tabular}
\par\smallskip
\begin{tabular}{@{}lrr@{}}\toprule
Quantity ($10^{-3}$) & Total budget & Position budgets\\\midrule
Budget term $\bm\lambda^\top\bm\varepsilon$ & $3.401$ & $1.184$\\
Envelope mean $\overline Q_{\bm\lambda}$ & $5.508$ & $2.806$\\
KL correction $U_{\bm\lambda}-\overline Q_{\bm\lambda}$ & $0.556$ & $0.403$\\
Envelope upper bound $U_{\bm\lambda}$ & $6.065$ & $3.209$\\
Excess-BLER upper bound & $9.466$ & $4.393$\\\bottomrule
\end{tabular}
\end{table}

{The mixed allocations in \mbox{Table~\ref{tab_onn_parameter_update}} have total budgets between $1.389\times10^{-4}$ and $1.524\times10^{-4}$. For the last allocation, \mbox{Table~\ref{tab_position_diagnostics}} gives a common multiplier $24.487$, versus position multipliers near $250$ at the first three positions and $5.514$ at the last. Each balances budget cost against low-margin risk. Separate constraints retain the larger budget at its assigned last position.}

{Budget, empirical-envelope and KL terms contribute $43.7\%$, $53.3\%$ and $3.0\%$ of the reduction, respectively. The position-resolved bound meets tolerance $0.005$. Across the five allocations, upper-bound tightening $10\log_{10}(U_{\rm total}/U_{\rm position})$ ranges from $0.128$ to $3.334$~dB for the same receivers, admitting more precision choices under the fixed allowance.}
\subsubsection{Exhaustive Assessment of Local Choices}
\label{sec_enumeration}
{At each of the four information positions, the $32$-bit floating-point center is combined with quantized versions whose parameter precisions are $\{20\}$, $\{20,16\}$, $\{20,16,12\}$ or $\{20,16,12,8\}$ bits. The four nested families contain $16,81,256,625$ decoders, respectively, assessed on a separate $2\!\times\!10^5$-frame sample. Each method splits $0.05$ across four families and paired tails. Common and exhaustive bounds coincide at $1.269\!\times\!10^{-3}$ for the first three families and $1.285\!\times\!10^{-3}$ for the fourth, versus Bonferroni bounds $[1.392,1.538]\times10^{-3}$. Empirical common-event gaps are zero.}
\vspace{-9pt}\subsection{\texorpdfstring{$(128,64)$ Reliability and Implementation Performance}{(128,64) Reliability and Implementation Performance}}\vspace{-3pt}
\label{subsec_n128_deployment_audit}

\subsubsection{Graph Decoder and Evaluation Protocol}
\label{subsec_n128_model_protocol}
{Each $(128,64)$ predictor incorporates polar parity and equality constraints through $T=17$ graph iterations, width-$32$ refiners and a width-$512$, two-block residual readout, with $1.185353\!\times\!10^6$ parameters. For $\bm x=\bm u\mathbf F^{\otimes7}$, one-based information positions are $31$--$32$, $44$, $46$--$48$, $52$, $54$--$56$, $58$--$64$, $72$, $76$, $78$--$80$, $84$, $86$--$96$, and $99$--$128$. Transmitted-prefix errors enter the common certification interface.}

For variable $v$ with neighboring factors $\mathcal N(v)$ and external LLR evidence $b_v$, denote variable-to-factor and factor-to-variable messages by $q_{v\to a}^{(t)}$ and $m_{a\to v}^{(t)}$. Starting from $m_{a\to v}^{(0)}=0$, parallel updates for $t=0,\ldots,T-1$ are
\begin{equation}
\begin{gathered}
q_{v\to a}^{(t)}=b_v+\sum_{c\in\mathcal N(v)\setminus\{a\}}m_{c\to v}^{(t)},\\
\begin{aligned}
m_{a\to v}^{(t+1)}&=\gamma m_{a\to v}^{(t)}+(1-\gamma)s\Bigl[B_{a\to v}(\bm q_a^{(t)})\\
&\hspace{5em}+\phi_{\bm\theta_{\tau(a)}}(\bm\xi_{a\to v}^{(t)})\Bigr].
\end{aligned}
\end{gathered}
\label{eq_n128_graph_update}
\end{equation}
{The vector $\bm q_a^{(t)}$ collects incoming messages. Excluding destination $v$, $B_{a\to v}$ applies parity box-plus $x\boxplus y=2\operatorname{atanh}(\tanh(x/2)\tanh(y/2))$ or copies the other equality message. The product is clipped to $\pm(1-2^{-23})$ in $32$-bit arithmetic, also used for transmitted-prefix SC diagnostics. Sequential SC, paired BLER certificates and teachers use min-sum. Type $\tau(a)$ selects one of two shared single-hidden-layer refiners with sigmoid linear unit~\mbox{(SiLU)} activations. Each position model shares trainable scale $s=\exp(\operatorname{clip}(\beta,-\log2,\log2))$ and damping $\gamma=\tfrac12\operatorname{sigmoid}(\eta)$ across factors and iterations, initially $\beta=0$, $\eta=-6$. They regulate update magnitude and persistence. Parity features are $(x,y,|x|,|y|,z,|z|,a(x),a(y))$ for raw update $z$, with $a(x)=\operatorname{sign}(x)\log(1+|x|)$. Equality features are $(z,|z|,a(z))$. Both append iteration then stage embeddings, each of dimension four, shared across factor types.}

Inputs comprise channel LLRs at the codeword layer and source evidence: $12(1-2u_j)$ for prefix bits, $+12$ for zero frozen bits and zero otherwise. Reference and sequential modes use transmitted and neural prefixes, respectively. Let $h_v^{(t)}=b_v+\sum_{a\in\mathcal N(v)}m_{a\to v}^{(t)}$ denote a variable belief, $\bm h_{\rm src}^{(T)}$ collect final source beliefs, and $\bm d_i=(h_i^{(T-2)},h_i^{(T-1)},h_i^{(T)})^\top$. An affine history readout $g_i$ and SiLU residual readout $G_i$ give
\begin{equation}
\begin{gathered}
\widehat\ell_i=g_i(\bm d_i)+G_i\!\left(\tanh\!\left(\frac{[\bm h_{\rm src}^{(T)};\bm\ell;\bm d_i]}{8}\right)\right),\\
F_i(\bm Z_i)=-\tanh(\widehat\ell_i/2).
\end{gathered}
\label{eq_n128_graph_readout}
\end{equation}
{The residual branch corrects the three-belief history readout using the normalized $259$-dimensional vector of source beliefs, channel LLRs and history. The output follows the $2u_i-1$ convention.}

Both training stages use a min-sum SC teacher. For teacher LLR $t$, prediction $p$, and $e=a(p)-a(t)$, the batch mean of the following loss is minimized:
\begin{equation}
\ell(p,t)=\operatorname{Huber}_1(e)+0.25\log(1+\mathrm{e}^{-\operatorname{sgn}_+(t)p}),
\end{equation}
The Huber term equals $e^2/2$ when $|e|\le1$ and $|e|-1/2$ otherwise. The teacher sign is $\operatorname{sgn}_+(t)=1$ for $t\ge0$ and $-1$ otherwise. The signed-log transform compresses LLR magnitudes. Huber regression limits the influence of large residuals, and the logistic term promotes teacher-sign agreement. Adam uses rate $0.002$, moments $(0.9,0.999)$, offset $10^{-8}$, batch $256$, and gradient clipping at one. Physical training levels $-3.01,-2.01,-1.01,-0.01,0.99$~dB are shuffled every five updates.

{The pretrained graph decoder comprises the graph and affine history readout after $3000$ updates. A residual readout is then added for $1000$ joint updates with reinitialized Adam state. Minimum loss on $2048$ fixed frames at $-1.01$~dB selects initialization or each $50$th refinement step. The residual branch maps its input to $\bm h=\operatorname{SiLU}(\mathbf W_0\bm x+\bm b_0)$, followed by two blocks $\bm h\leftarrow\bm h+\mathbf W_2\operatorname{SiLU}(\mathbf W_1\bm h+\bm b_1)+\bm b_2$. Correction outputs, the second affine map in each residual block and the final output map start at zero. The history readout initially selects the final belief.}

Evaluation and certification samples are generated independently of training and validation.

\subsubsection{BER and BLER Results}
\label{subsec_long_results_v2}

The neural, SC, SCL-8 and pretrained graph decoder curves use common frames in the AWGN comparisons of \mbox{Figs.~{\hyperref[fig_long_code_performance_2x2_v2_a]{\ref*{fig_long_code_performance_2x2_v2}(a)}} and~{\hyperref[fig_long_code_performance_2x2_v2_b]{\ref*{fig_long_code_performance_2x2_v2}(b)}}} and the block-Rayleigh comparisons of \mbox{Figs.~{\hyperref[fig_long_code_performance_2x2_v2_c]{\ref*{fig_long_code_performance_2x2_v2}(c)}} and~{\hyperref[fig_long_code_performance_2x2_v2_d]{\ref*{fig_long_code_performance_2x2_v2}(d)}}}. SC and SCL-8 use min-sum updates, with max-log paths and no CRC for SCL-8. AWGN batches at $0$--$4$~dB continue until each decoder has $100$ block errors. Rayleigh points at $0$--$10$~dB use $10^{4}$ frames. These curves are descriptive. The mixed $4/8$-bit decoder in \mbox{Figs.~{\hyperref[fig_long_code_performance_2x2_v2_b]{\ref*{fig_long_code_performance_2x2_v2}(b)}} and~{\hyperref[fig_long_code_performance_2x2_v2_d]{\ref*{fig_long_code_performance_2x2_v2}(d)}}} uses independent samples.

\begin{figure*}[t]
\centering
\captionsetup[subfloat]{farskip=0pt,nearskip=0pt,captionskip=1pt}
\setlength{\abovecaptionskip}{2pt}
\subfloat[AWGN BER\label{fig_long_code_performance_2x2_v2_a}]{\includegraphics[width=2.50in]{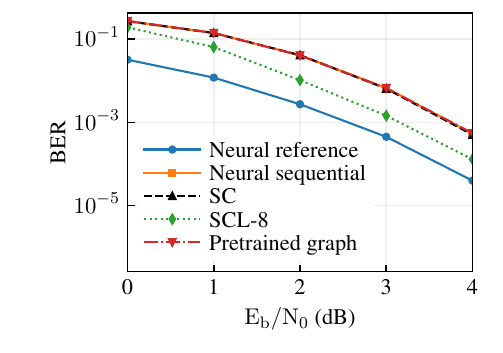}}\hspace{0.01in}
\subfloat[AWGN BLER\label{fig_long_code_performance_2x2_v2_b}]{\includegraphics[width=2.50in]{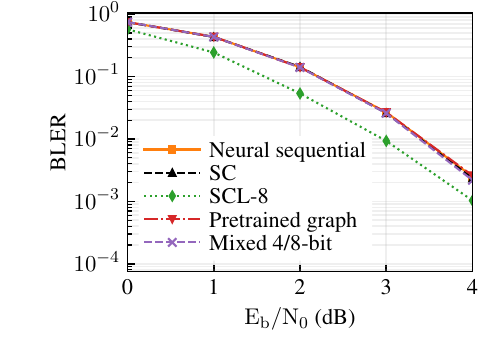}}\\[1.2ex]
\subfloat[Block-Rayleigh BER\label{fig_long_code_performance_2x2_v2_c}]{\includegraphics[width=2.50in]{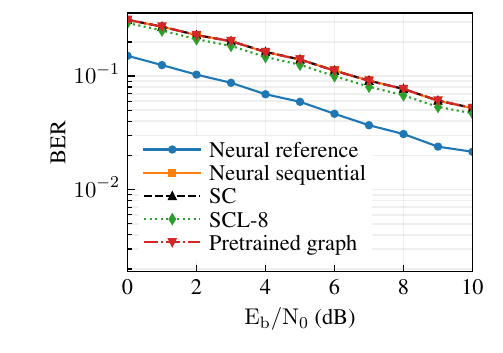}}\hspace{0.01in}
\subfloat[Block-Rayleigh BLER\label{fig_long_code_performance_2x2_v2_d}]{\includegraphics[width=2.50in]{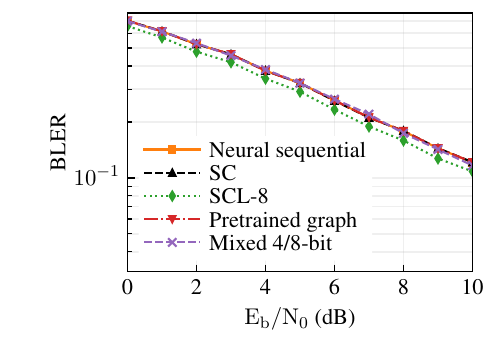}}
\caption{BER and BLER for the $(128,64)$ polar code under AWGN and block-Rayleigh fading.}
\label{fig_long_code_performance_2x2_v2}
\end{figure*}

On the same AWGN frames, refinement corrects $57,58,27,13,17$ pretrained-decoder failures and creates $69,56,32,16,8$ new failures at the five respective points. This paired comparison measures the complete refinement, including the readout and graph updates.

The mixed $4/8$-bit decoder removes the residual readout after joint refinement and quantizes $41$ predictors to $4$ bits per parameter and $23$ to $8$ bits. Its BLER uses independent $5\!\times\!10^{4}$-frame certification samples at AWGN $3$ and $4$~dB and Rayleigh $4$, $6$ and $8$~dB, with $10^{4}$ new frames at other points. The first-error identity gives exact sequential BLER from transmitted-prefix scores.

\subsubsection{Certified Readout Simplification}
\label{subsec_long_certificate_v2}
Readout simplification is assessed through two degrees of parameter compression and module removal. Four local implementations retain the $32$-bit floating-point readout, quantize it to $16$ or $8$ bits per parameter, or remove it. Quantization applies the preceding symmetric rule to each affine output row including bias, with $64$-bit floating-point scales and $32$-bit floating-point dequantization. Removal retains the jointly refined graph and history readout.

At AWGN $4$~dB (A4) and block-Rayleigh $6$~dB (R6), six bounds on harmful events and absolute disagreement for the center, $8$-bit and pruned banks, plus shared-frame-event and radius bounds, give $16$ statements across the two laws, each at $3.125\!\times\!10^{-3}$. The $16$-bit option enters the local family. Common harmful counts $(18,99)$ yield bounds $(6.62\!\times\!10^{-4},\ 2.59\!\times\!10^{-3})$ at A4 and R6, respectively. Both meet allowances $0.001$ and $0.005$ for all four-choice assemblies. The pruned bank retains $1096$ parameters per position.

\subsubsection{Precision Choices for the Retained Graph}
\label{sec_graph_precision}
The $8$, $4$ and $2$-bit choices successively halve the nominal storage per parameter. Nested libraries add these quantized implementations to a $32$-bit floating-point control to assess how expanded storage choices affect the common reliability bound. Each tensor uses one $64$-bit floating-point scale and the preceding rule, with graph connectivity preserved. The graph study in \mbox{Table~\ref{tab_protocol_index}} uses the A4 and R6 operating conditions. Harmful, paired binomial, and paired Bernstein bounds cover four banks and three nested families. These bounds under both laws give 42 statements, each assigned failure probability $\delta\approx1.1905\!\times\!10^{-3}$.

{At A4 and R6, respectively, paired binomial bounds are $(8.36\!\times\!10^{-4},2.084\!\times\!10^{-3})$ for $\{32,8\}$, $(2.712\!\times\!10^{-3},1.5878\!\times\!10^{-2})$ for $\{32,8,4\}$, and $(0.970877,0.702251)$ for $\{32,8,4,2\}$. Uniform $8$-bit parameters give $(8.35\!\times\!10^{-4},1.749\!\times\!10^{-3})$. The $\{32,8\}$ family meets allowances $0.001$ and $0.005$, equivalent to $47.2\%$ and $1.90\%$ of the SC empirical BLERs $(2.12\!\times\!10^{-3},0.26310)$. Larger libraries require higher allowances. SCL-8 provides the list-search reference in \mbox{Fig.~\ref{fig_long_code_performance_2x2_v2}}.}

\subsubsection{Candidate Scale, Error Overlap, and Sample Size}
\label{sec_assembly_scaling}
For $(128,64)$, two families vary parameter precision at $m\in\{1,4,8,12\}$ selected information positions. One permits $32$-bit floating-point or $8$-bit quantized parameters at each position, and the other permits $8$-bit or $4$-bit quantized parameters. The remaining predictors use $8$-bit quantized parameters. \mbox{Table~\ref{tab_strong_baselines}} reports $(m,M)=(8,2\!\times\!10^4)$ and $(12,5\!\times\!10^4)$, with $256$ and $4096$ decoder combinations, respectively. Here $M$ is the certification frame count per channel. All $2^m$ combinations are assessed exhaustively. The $5\!\times\!10^4$ frames per law from \mbox{Section~\ref{sec_graph_precision}} rank positions by decreasing local disagreement, with ties by information rank. The higher precision defines the center. Two channels, two precision pairs, four $m$ values and three sample sizes $M\in\{5000,2\!\times\!10^4,5\!\times\!10^4\}$ give $48$ settings per method. All methods reuse $64+m$ local predictions per frame. The predetermined bank uses lower precision at the selected positions.

Fixed-sequence tests order combinations by independent design error counts and stop at the first failure. Each method uses per-setting failure probability $\delta\approx1.0417\!\times\!10^{-3}$. Memberwise Bonferroni testing further divides by $2^m$.

\begin{table*}[t]\centering\footnotesize
\caption{BLER-difference upper bounds for $(128,64)$ decoder families.}
\label{tab_strong_baselines}
\setlength{\tabcolsep}{1.5pt}
\begin{tabular*}{\textwidth}{@{}>{\raggedright\arraybackslash}m{0.12\textwidth}@{\extracolsep{\fill}}>{\centering\arraybackslash}m{0.065\textwidth}>{\centering\arraybackslash}m{0.14\textwidth}>{\centering\arraybackslash}m{0.075\textwidth}>{\centering\arraybackslash}m{0.10\textwidth}*{3}{>{\centering\arraybackslash}m{0.15\textwidth}}@{}}\toprule
\multicolumn{8}{@{}l@{}}{Upper bound on maximum excess BLER relative to SC over all $2^m$ decoder combinations in each displayed setting.}\\
\multicolumn{8}{@{}l@{}}{Joint confidence: $\ge95\%$ across $48$ settings, separately for each method. Per-setting failure probability: $\delta\approx1.0417\!\times\!10^{-3}$.}\\
\multicolumn{8}{@{}l@{}}{Full study: $m\in\{1,4,8,12\}$, $M\in\{5000,2\!\times\!10^4,5\!\times\!10^4\}$. Displayed: $(m,M)=(8,2\!\times\!10^4)$ and $(12,5\!\times\!10^4)$.}\\
\midrule
Channel & \shortstack{$\mathrm E_{\mathrm b}/\mathrm N_0$\\(dB)} & \shortstack{Parameter precision\\(bits per parameter)} & \shortstack{Variable\\positions\\$m$} & \shortstack{Frames per\\channel\\$M$} & \shortstack{Common-event\\bound} & Verified LP & \shortstack{Exhaustive\\enumeration}\\\midrule
AWGN & $4$ & $\{32,8\}$ & $8$ & $2\!\times\!10^4$ & $9.18\!\times\!10^{-4}$ & $9.18\!\times\!10^{-4}$ & $9.18\!\times\!10^{-4}$\\
AWGN & $4$ & $\{8,4\}$ & $8$ & $2\!\times\!10^4$ & $3.182\!\times\!10^{-3}$ & $3.119\!\times\!10^{-3}$ & $3.048\!\times\!10^{-3}$\\
Block-Rayleigh & $6$ & $\{32,8\}$ & $8$ & $2\!\times\!10^4$ & $2.627\!\times\!10^{-3}$ & $2.530\!\times\!10^{-3}$ & $2.530\!\times\!10^{-3}$\\
Block-Rayleigh & $6$ & $\{8,4\}$ & $8$ & $2\!\times\!10^4$ & $1.5075\!\times\!10^{-2}$ & $1.3816\!\times\!10^{-2}$ & $1.1133\!\times\!10^{-2}$\\\midrule
AWGN & $4$ & $\{32,8\}$ & $12$ & $5\!\times\!10^4$ & $6.22\!\times\!10^{-4}$ & $6.22\!\times\!10^{-4}$ & $6.22\!\times\!10^{-4}$\\
AWGN & $4$ & $\{8,4\}$ & $12$ & $5\!\times\!10^4$ & $2.731\!\times\!10^{-3}$ & $2.647\!\times\!10^{-3}$ & $2.542\!\times\!10^{-3}$\\
Block-Rayleigh & $6$ & $\{32,8\}$ & $12$ & $5\!\times\!10^4$ & $1.531\!\times\!10^{-3}$ & $1.468\!\times\!10^{-3}$ & $1.421\!\times\!10^{-3}$\\
Block-Rayleigh & $6$ & $\{8,4\}$ & $12$ & $5\!\times\!10^4$ & $1.4377\!\times\!10^{-2}$ & $1.2901\!\times\!10^{-2}$ & $9.912\!\times\!10^{-3}$\\\bottomrule
\end{tabular*}
\end{table*}
Verified LP tightens $26$ of $48$ $(128,64)$ bounds and certifies the same $21$ families as exhaustive enumeration. Across both codes, all $96$ LP bounds dominate the enumerated maxima. Tests on $240$ random local tables verify the dual and adaptive-envelope inequalities. Two LPs bound harmful and SC-error counts after identical error patterns are grouped. Dual numerators use exact-integer verification at denominator $2^{20}$. For $m=12$, seven-run median costs (common-event bound, verified LP, exhaustive enumeration) are $(0.24,7.84,66.83)$~ms for $(128,64)$ and $(0.21,8.08,70.87)$~ms for $(64,48)$. Costs include grouping, optimization, integer checks and inversion after shared inference and error-table construction.

For $\{8,4\}$ parameters under block-Rayleigh fading at $6$~dB, $m=12$, $M=5\!\times\!10^4$, the empirical common-event gap is $5.10\!\times\!10^{-3}$, with interval $[3.623\!\times\!10^{-3},6.937\!\times\!10^{-3}]$. These intervals separately allocate $\delta\,2^{-(m+1)}$ per tail for joint $95\%$ confidence. The positive gap supports combination-specific refinement.

\subsubsection{Selecting a Mixed-Precision Receiver}
{For allowances $0.001$ at A4 and $0.005$ at R6, the $5\!\times\!10^4$ frames per law from \mbox{Section~\ref{sec_graph_precision}} are reused for design to identify $41$ positions, denoted $\mathcal T_4$, where $4$-bit parameters add no harmful frame to the $\{32,8\}$ family. These positions admit $\{32,8,4\}$ parameters and the others $\{32,8\}$. The minimum-storage member uses $4$-bit parameters on $\mathcal T_4$ and $8$-bit parameters elsewhere. The $3^{41}2^{23}$ family is assessed on $5\!\times\!10^4$ fresh frames per law. Two paired methods assess the baseline family, mixed family and selected bank under both laws: $12$ bounds at $\delta\approx4.1667\!\times\!10^{-3}$ give joint $95\%$ coverage.}

Common paired binomial bounds $(5.63\!\times\!10^{-4},\ 1.776\!\times\!10^{-3})$ meet both targets. At A4 and R6, the selected bank has $109$ and $1.3268\!\times\!10^{4}$ block errors, respectively, versus SC's $104$ and $1.3263\!\times\!10^{4}$, with paired bounds $(4.81\!\times\!10^{-4},\ 1.254\!\times\!10^{-3})$. Family Bernstein bounds $(1.066\!\times\!10^{-3},\ 2.185\!\times\!10^{-3})$ exceed the A4 allowance. The comparison highlights the importance of the statistical method at this risk level.

\subsubsection{Storage Savings and Certification Cost}
\label{sec_mixed_deployment}
{For the selected $(128,64)$ receiver, $41$ predictors use $4$-bit parameters and $23$ use $8$-bit parameters. A shared tensor-layout description stores each parameter name and shape once across positions. In this format, uniform $8$-bit and mixed $4/8$-bit files occupy $1.05493\!\times\!10^5$ and $8.3025\!\times\!10^4$ bytes, respectively, a $21.30\%$ reduction. Their parameter payloads occupy $7.0144\!\times\!10^4$ and $4.7676\!\times\!10^4$ bytes, with the same $3.5349\!\times\!10^4$ bytes of scales, graph structure, metadata and container overhead. Both execute with dequantized $32$-bit floating-point parameters.}

The initial family assessment uses the confirmation sample in \mbox{Section~\ref{sec_budget_selection}}: $5\!\times\!10^4$ frames per channel. It takes $1406.54$~s, including file verification. Evaluating all $64$ predictors at each of $32$, $8$ and $4$ bits per parameter takes $1370.06$~s ($97.4\%$). The remaining $36.48$~s covers loading, sampling, SC, persistence, statistics, file verification and overhead. Candidate design precedes assessment. The resulting certificate supports subsequent feasible selections.

\subsubsection{Delayed Receiver-Update Constraints}
\label{sec_budget_selection}
{Each configuration selection task specifies a storage limit, writable predictor positions and channel-distortion weights for updating an installed $8$-bit decoder. Other parameters and buffers remain fixed.}

For writable mask $\mathcal W$, the eligible positions are $\mathcal T=\mathcal W\cap\mathcal T_4$. Each replacement with $4$-bit quantized parameters saves $548$ bytes, so byte budget $C_{\max}$ requires at least $k=\max\{0,\lceil(1.05493\!\times\!10^{5}-C_{\max})/548\rceil\}$ replacements. {On $5\!\times\!10^4$ independent design frames per law, $\widehat e_{i,b}(w)$ estimates mean-squared output change from $b$-bit parameters to the retained $32$-bit floating-point graph and history readout. Both use~\eqref{eq_n128_graph_readout} with the residual branch removed. A4 and R6 receive weights $w$ and $1-w$.} With $d_i=\widehat e_{i,4}-\widehat e_{i,8}$, select
\begin{equation}
\min_{\substack{\mathcal R\subseteq\mathcal T\\|\mathcal R|\ge k}}
\left[\sum_i\widehat e_{i,8}+\sum_{i\in\mathcal R}d_i\right].
\label{eq_destination_objective}
\end{equation}
Choose the lowest $\max\{k,|\{i\in\mathcal T:d_i<0\}|\}$ increments, ties by rank. If $k>|\mathcal T|$, the {task} is infeasible. The family certificate supplies reliability independently of this distortion objective.

{The $472$ feasible tasks select $256$ distinct decoders before confirmation. Forty tasks combine five weights $w\in\{0,0.25,0.5,0.75,1\}$ with eight feasible pairs: writable sets contain all eligible positions or those with even or odd information-position rank, and storage limits are $8.4\!\times\!10^4$, $9\!\times\!10^4$, $9.6\!\times\!10^4$ or $1.04\!\times\!10^5$ bytes. Another $432$ combine $16$ distinct uniformly sampled eligible subsets per size $r\in\{12,20,28\}$, weights $w\in\{0,0.5,1\}$ and limits $1.05493\!\times\!10^5-548k$ bytes for $k\in\{r,r-2,r-4\}$.} {On $5\!\times\!10^4$ new frames per law, full-family shared harmful and gain counts $(16,9)$ at A4 and $(99,78)$ at R6 give bounds $(4.79\times10^{-4},1.283\times10^{-3})$. Common and cached online tests certify all tasks. After shared inference, their seven-run median aggregation times are $5.60$ and $151.17$~ms. Online testing divides $0.05$ among tasks, laws and tails.}

At A4 and R6, paired binomial bounds are $(4.27\times10^{-4},8.77\times10^{-4})$ for the predetermined minimum-storage decoder. Their maxima over the $256$ {selected} decoders are $(4.27\times10^{-4},9.47\times10^{-4})$ by~\eqref{eq_maximum_family_bound}. For unknown writable sets, verified LP lowers the R6 family bound from $1.283\times10^{-3}$ to $1.097\times10^{-3}$ and certifies every feasible selection at allowances $(5\!\times\!10^{-4},1.1\!\times\!10^{-3})$. Each comparison allocates $0.025$ per law and $0.0125$ per tail for its own $95\%$ event.

Initial confirmation costs $1406.54/n_{\rm req}$ seconds per admissible {task} ($2.98$~s for $472$ {tasks}). Reusing its predictions, common and cached online aggregation differ by $0.14557$~s. The stored family certificate covers subsequent feasible selections within the same implementations and operating laws.

\subsubsection{Refining Common Events}
\label{sec_partition_experiment}
The $(128,64)$ families in \mbox{Table~\ref{tab_partition_results}} use the $12$ positions and their order from \mbox{Section~\ref{sec_assembly_scaling}}, fixed using the independent design samples of \mbox{Section~\ref{sec_graph_precision}}. Assessment reuses the $M=5\!\times\!10^4$ selection frames per channel from \mbox{Section~\ref{sec_budget_selection}}. Precision is fixed at the first $d$ of $m=12$ variable information positions, producing $2^d$ subsets of $2^{12-d}$ decoders. Their maximum bound covers all $4096$ combinations. At $d=0$, the formula gives the common-event bound. Each of the four channel--precision settings receives failure probability $1.25\!\times\!10^{-2}$. The joint guarantee covers all settings and assessed depths.

\begin{table*}[!t]\centering\footnotesize
\caption{BLER-difference upper bounds under partition refinement for $(128,64)$ decoder families.}
\label{tab_partition_results}
\setlength{\tabcolsep}{1.5pt}
\begin{tabular*}{\textwidth}{@{}>{\raggedright\arraybackslash}m{0.14\textwidth}@{\extracolsep{\fill}}>{\centering\arraybackslash}m{0.09\textwidth}>{\centering\arraybackslash}m{0.19\textwidth}*{3}{>{\centering\arraybackslash}m{0.18\textwidth}}@{}}\toprule
\multicolumn{6}{@{}l@{}}{Upper bound on maximum excess BLER relative to SC over all $2^m=4096$ decoder combinations.}\\
\multicolumn{6}{@{}l@{}}{Joint confidence: $\ge95\%$ across $4$ settings and all assessed depths. Per-setting failure probability: $1.25\!\times\!10^{-2}$.}\\
\multicolumn{6}{@{}l@{}}{Variable positions: $m=12$. Frames per channel: $M=5\!\times\!10^4$. Other $64-m$ predictors: $8$-bit quantized parameters.}\\\midrule
Channel & \shortstack{$\mathrm E_{\mathrm b}/\mathrm N_0$\\(dB)} & \shortstack{Parameter precision\\(bits per parameter)} & \shortstack{$d=0$\\1 subset} & \shortstack{$d=4$\\16 subsets} & \shortstack{$d=12$\\4096 subsets}\\\midrule
AWGN & $4$ & $\{32,8\}$ & $6.18\!\times\!10^{-4}$ & $6.18\!\times\!10^{-4}$ & $6.10\!\times\!10^{-4}$\\
AWGN & $4$ & $\{8,4\}$ & $2.414\!\times\!10^{-3}$ & $2.117\!\times\!10^{-3}$ & $2.066\!\times\!10^{-3}$\\
Block-Rayleigh & $6$ & $\{32,8\}$ & $1.545\!\times\!10^{-3}$ & $1.500\!\times\!10^{-3}$ & $1.381\!\times\!10^{-3}$\\
Block-Rayleigh & $6$ & $\{8,4\}$ & $1.4188\!\times\!10^{-2}$ & $1.1279\!\times\!10^{-2}$ & $9.385\!\times\!10^{-3}$\\\bottomrule
\end{tabular*}
\end{table*}

Finer partitions tighten the guarantee, particularly for $\{8,4\}$ parameters under block-Rayleigh fading. The $\{32,8\}$ families meet BLER-difference allowances $0.001$ (AWGN, $4$~dB) and $0.005$ (block-Rayleigh, $6$~dB). The same event permits assessment-dependent depth selection and combination with verified LP bounds.

\subsubsection{Validation at Additional Operating Points}
Using $5\!\times\!10^{4}$ new frames per law and six bounds at failure probability $\delta\approx8.3333\!\times\!10^{-3}$ each, the frozen mixed family also meets an excess-BLER allowance $0.005$ at AWGN $3$~dB and Rayleigh $4$ and $8$~dB. The bounds $(1.755,2.112,1.000)\times10^{-3}$ certify the same implementations at these three additional operating laws.

\vspace{-9pt}\subsection{\texorpdfstring{Independent $(64,48)$ Evaluation}{Independent (64,48) Evaluation}}\vspace{-3pt}
\label{sec_n64_independent}
{An independent $(64,48)$ bank has frozen positions $1$--$7$, $9$--$11$, $13$, $17$--$19$, $33$ and $34$. Each $1092$-parameter predictor uses $17$ iterations, width-$32$ refiners, four-dimensional embeddings and an affine history readout. The $(128,64)$ initialization, loss, Adam settings and training levels are retained. Training runs for $3000$ updates, with checkpoint selection on $2048$ fixed frames at $-1.01$~dB, initially and every $100$ steps.}

{Independent design uses $10^4$ frames per law to fix position rankings, $32$-, $8$- and $4$-bit banks, families and the selected decoder. Certification uses $5\!\times\!10^4$ new frames per law at $\delta\approx4.1667\!\times\!10^{-3}$ per bound. Size comparisons reuse these frames at $\delta\approx1.0417\!\times\!10^{-3}$ per setting and method. The mixed family admits $4$-bit parameters at $35$ positions. Its bounds $(2.356,2.269)\times10^{-3}$ meet allowance $0.005$ under both laws, whereas all-$4$ bounds $(7.234,8.050)\times10^{-3}$ fail. The selected file saves $26.72\%$ over all-$8$ in the same format. At $m=12$, $M=5\!\times\!10^4$, the $\{8,4\}$ common-event gap intervals are $[2.989,6.051]\times10^{-3}$ and $[4.044,7.515]\times10^{-3}$, confirming positive conservatism at this second code configuration.}

\vspace{-0.08in}\vspace{-6pt}\section{Implementation Cost}\vspace{-3pt}
\label{sec_complexity}

An ONN bank requires $\mathcal O(KBN)$ decoding arithmetic and storage, and $\mathcal O(DKBN)$ arithmetic per full-batch step.

{Let $n_i=|\mathcal G_i|$ be the number of implementations at position $i$. Certification computes $M\sum_i n_i$ local outputs once. Shared-event aggregation requires $\mathcal O(M\sum_i n_i)$ Boolean operations. With the same cached outputs, explicit assessment of $n_{\rm cfg}$ combinations requires $\mathcal O(n_{\rm cfg}MK)$ Boolean operations, where exhaustive assessment has $n_{\rm cfg}=\prod_i n_i$. The shared bound trades a potentially conservative envelope for aggregation independent of the combination count.}

For $T$ iterations, refiner width $h$, readout width $H$, and $L$ residual blocks, $(128,64)$ graph-bank arithmetic is $\mathcal O(K[TN\log N\,h+NH+LH^2])$. Shared refiners require storage $\mathcal O(K[h+T+\log N+NH+LH^2])$. 

Full and pruned banks contain $7.5862592\!\times\!10^{7}$ and $7.0144\!\times\!10^{4}$ parameters. Resident GPU allocations after model loading are $291.28$ and $2.12$~mebibytes~\mbox{(MiB)}, with inference peaks of $581.55$ and $289.02$~MiB, respectively. Five synchronized $512$-frame batches after warm-up on identical AWGN $4$~dB frames yield median batch-throughput costs $4.433$ and $4.401$~ms/frame, with ranges $4.413$--$4.459$ and $4.392$--$4.404$. Python SC and SCL-8 take $0.5$--$1.4$ and $1.2$--$3.2$~s per $10^{4}$ frames, respectively.

{A separate instrumented run uses five $512$-frame AWGN $4$~dB batches after two warm-up batches. Both runs measure synchronized wall-clock time with sequential decision prefixes. Instrumented GPU events measure component intervals, including launch gaps.} Its median full and pruned batch times are $2.333$ and $2.330$~s. Input preparation and graph propagation account for $98.7\%$ and $99.1\%$, respectively. The full model's residual readout accounts for $0.36\%$. Graph operations repeat $17$ times per information position, whereas the residual readout acts once. Readout removal therefore preserves the dominant computation.

\vspace{-0.08in}\vspace{-6pt}\section{Conclusion}\vspace{-3pt}
\label{sec_conclusion}
{The framework connects conditional network fitting, independent population-risk certification and sequential reliability. Position-specific output-change budgets protect weight updates. Shared frame losses certify fixed implementation combinations, with an exact population gap from the worst fixed decoder. Verified LP and adaptive partitioning tighten these guarantees at fixed confidence. ONN and factor-graph experiments support certified update and precision selection across code configurations. Resource measurements quantify storage reductions and the dominant graph-propagation cost.}

\appendices
\renewcommand{\thetheorem}{\Alph{section}.\arabic{theorem}}
\renewcommand{\theremark}{\Alph{section}.\arabic{remark}}
\renewcommand{\theproposition}{\Alph{section}.\arabic{proposition}}
\renewcommand{\theassumption}{\Alph{section}.\arabic{assumption}}
\renewcommand{\thelemma}{\Alph{section}.\arabic{lemma}}
\renewcommand{\thecorollary}{\Alph{section}.\arabic{corollary}}
\renewcommand\thefigure{A.\arabic{figure}}
\renewcommand\thetable{A.\arabic{table}}
\setcounter{table}{0}
\setcounter{figure}{0}
\vspace{-6pt}\section{Derivation of the Empirical Bernstein Certificates}\vspace{-3pt}
\label{app_empirical_bernstein}
For i.i.d. losses $L_m\in[0,4]$, define $X_m=L_m/4$. Then $\widehat R_X=\widehat R_L/4$ and $\widehat V_X=\widehat V_L/16$. The unit-interval radius~\mbox{\cite{MaurerPontil2009}} is $\sqrt{2\widehat V_X\log(2/\delta)/M}+7\log(2/\delta)/[3(M-1)]$ at confidence $1-\delta$. Rescaling yields~\eqref{eq_heldout_empirical_bernstein}, also for frame averages.
For $X\in[-1,1]$, the transform $(X+1)/2$ gives upper bound $\bar X+\sqrt{2\widehat V_X\log(2/\delta)/M}+14\log(2/\delta)/[3(M-1)]$, capped at one. This bound applies to the {assembly envelope} in~\eqref{eq_signed_assembly}.

\vspace{-4pt}\vspace{-6pt}\section{Optimization Characterization of the Margin Bound}\vspace{-3pt}
\label{app_margin_profile}
Suppress the information-position index. If $U\ge\mathbb E[\gamma^2]$, choosing $a=1$ and $\lambda=0$ gives primal and dual value $\mathbb E[\gamma]$. If $U=0$, feasible $a$ vanishes at positive margins and both values are zero. For $0<U<\mathbb E[\gamma^2]$, choose $t\in(0,1]$ such that
$\mathbb E[\gamma^2\mathbb I(\gamma<t)]\le U\le\mathbb E[\gamma^2\mathbb I(\gamma\le t)]$.
Such a threshold $t$ is a quantile of the measure weighted by $\gamma^2$. Set $a=1$ below $t$ and zero above it, with a fractional value at $t$ chosen to exhaust the budget. With $\lambda=1/t$,
$\mathbb E[\gamma a(\gamma)]=\lambda U+\mathbb E[(\gamma-\lambda\gamma^2)_+]$.
Equality with the weak-duality bound proves~\eqref{eq_margin_profile}. Threshold randomization represents conditional disagreement at an atom, so the result also applies to discrete margin distributions.

\vspace{-4pt}\vspace{-6pt}\section{Duality for Position-Resolved Budgets}\vspace{-3pt}
\label{app_position_budget}
Let $\mathcal A$ be the set of measurable allocations $a_i\ge0$ with $\sum_i a_i\le1$ almost surely. The pointwise simplex constraints are weak-star closed. Thus $\mathcal A$ is a weak-star compact convex subset of the unit ball in a finite product of $L^\infty$ spaces. For fixed $\bm\lambda\ge\bm0$, define
\begin{equation}
L(\bm a,\bm\lambda)=\bm\lambda^\top\bm\varepsilon+
\mathbb E\left[C\sum_i(1-\lambda_i f_i^2)a_i\right].
\end{equation}
The bounded coefficients make $L$ continuous and affine in each variable. Sion's minimax theorem~\mbox{\cite{sion1958minimax}} gives
$\sup_{\bm a\in\mathcal A}\inf_{\bm\lambda\ge\bm0}L
=\inf_{\bm\lambda\ge\bm0}\sup_{\bm a\in\mathcal A}L$.
For a feasible allocation, the left inner infimum is $\mathbb E[C\sum_i a_i]$ at $\bm\lambda=\bm0$. Otherwise it is $-\infty$. The right inner supremum assigns unit mass to the smallest index that minimizes $\lambda_i f_i^2$ when this minimum is below one, and zero mass otherwise. This allocation is measurable. Its value is $\bm\lambda^\top\bm\varepsilon+\mathbb E[C(1-\min_i\lambda_i f_i^2)_+]$. Adding $\mathbb E J$ proves~\eqref{eq_position_primal}. The argument covers zero budgets and margins, including cases with unbounded minimizing sequences. The simplex encodes mutually exclusive first errors. Standard minimax theory evaluates the resulting decoder-specific relaxation.
{}

\vspace{-4pt}\vspace{-6pt}\section{Bounded-Mean Chernoff Inversion}\vspace{-3pt}
\label{app_bounded_kl}
For independent $X_m\in[0,1]$ with common mean $\mu$, convexity gives $\mathbb E[\mathrm{e}^{tX_m}]\le1-\mu+\mu \mathrm{e}^t$. Applying Markov's inequality to $\mathrm{e}^{t\sum_mX_m}$ for $t<0$ and minimizing over $t$ yields
$\mathbb P(\bar X\le x)\le \mathrm{e}^{-M\operatorname{KL}(x\|\mu)}$ for $0\le x<\mu$.
Inverting this lower-tail test proves $\mu\le U_{\rm KL}(\bar X,M,\delta)$ with probability at least $1-\delta$. The endpoint values are $1-\delta^{1/M}$ at $\bar X=0$ and $1$ at $\bar X=1$.

\vspace{-0.14in}{

}

\vfill

\end{document}